\documentclass[11pt]{article}

\input{style.tex}

\usepackage{amsmath, amssymb, amstext, amsfonts, mathrsfs}

\begin{document}

\title{\vspace*{-2.9cm} A two-player portfolio tracking 
  game}
 
\author{ Moritz Vo{\ss}\footnote{University of California Los Angeles,
    Department of Mathematics, Los Angeles, CA 90095, USA, email \texttt{voss@math.ucla.edu}.}  }
\date{\today}

\maketitle
\begin{abstract}
We study the competition of two strategic agents for liquidity in the benchmark portfolio tracking setup of~\citet{BankSonerVoss:17}. Specifically, both agents track their own stochastic running trading targets while interacting through common aggregated temporary and permanent price impact \`a la~\citet{AlmgChr:01}. The resulting stochastic linear quadratic differential game with terminal state constraints allows for a unique and explicitly available open-loop Nash equilibrium. Our results reveal how the equilibrium strategies of the two players take into account the other agent's trading targets: either in an exploitative intent or by providing liquidity to the competitor, depending on the relation between temporary and permanent price impact. As a consequence, different behavioral patterns can emerge as optimal in equilibrium. These insights complement and extend existing studies in the literature on predatory trading models examined in the context of optimal portfolio liquidation games.
\end{abstract}
 
\begin{description}
\item[Mathematical Subject Classification (2010):] 91A15, 91A23,
  91G80,\\ 49N10, 60H30
\item[JEL Classification:] C61, C73, G11
\item[Keywords:] stochastic differential game, Nash equilibrium,
  illiquid markets, portfolio tracking, predatory trading 
\end{description}


\section{Introduction}

In recent years, studying so-called \emph{price impact games} (also referred to as \emph{market impact games}) in the context of optimal portfolio liquidation problems has gained a lot of attraction in the financial mathematics literature. They investigate the strategic interaction of financial agents, who simultaneously trade in the same risky asset in order to cost-efficiently liquidate their position while affecting the asset's execution price through jointly generated price impact. That is, influencing the price in an adverse manner when they execute their buy or sell orders. These price impact games provide a tractable way to formalize the competition between agents for a risky asset's liquidity. Among the first game-theoretic approaches carried out to investigate possible phenomena in a competitive equilibrium where agents seek to liquidate their positions in the same risky asset are, e.g.,~\citet{BrunnermeierPedersen:05},~\citet{AttariMelloRuckes:05},~\citet{CarlinLoboViswanathan:07},~\citet{Schoen:08},~\citet{SchiedSchoeneborn:09},~\citet{CarmonaYang:09}, and~\citet{SchiedZhang:17}.

Our goal in this paper is to extend these works by formulating and studying the competition between two strategic agents for liquidity when both agents are trading simultaneously in an illiquid risky asset affected by price impact, because each agent seeks to \emph{track} her own exogenously given stochastic target strategy like, e.g., a frictionless delta hedge to dynamically hedge the fluctuations of their random endowments. Single-agent optimal tracking problems in the presence of price impact have first been considered by \citet{RogerSin:10}, \citet{NaujWes:11}, \citet{HorstNaujokat:14}, and~\citet{CarteaJaimungal16}. To the best of our knowledge, the present manuscript is the first to study a dynamic tracking problem in a competitive two-player price impact game setting. Specifically, we extend the single-player cost optimal benchmark portfolio tracking problem studied in~\citet{BankSonerVoss:17} in the presence of temporary and permanent price impact as proposed by~\citet{AlmgChr:01} to a two-player stochastic differential game. Both strategic agents are fully aware of the opponent's individual tracking objectives and they compete for available liquidity as the jointly caused price impact on the execution price directly feeds into their trading performances. We also allow for individual stochastic terminal state constraints on each agent's final portfolio position. Our aim is to shed light on the strategic interplay between the agents and to make transparent how each agent takes into account the other agent's trading targets in an optimal cost minimizing manner by solving for a Nash equilibrium in this two-player price impact game.
 
The paper most closely related to ours is~\citet{SchiedZhang:17}. Therein, the authors determine a unique open-loop Nash equilibrium within the class of deterministic strategies of agents aiming to liquidate a given asset position by maximizing a mean-variance criterion in an~\citet{AlmgChr:01} framework. Their study is an extension of the corresponding deterministic differential game solved in~\citet{CarlinLoboViswanathan:07} of liquidating risk-neutral agents who maximize expected revenues. Other extensions of the latter game include, e.g.,~\citet{SchiedSchoeneborn:09},~\citet{CarmonaYang:09},~\citet{MoallemiParkVanRoy:12},~\citet{ChuLehnertPassmore:09}. In contrast to these papers, which focus on optimal portfolio liquidation only, we additionally allow the agents to track their own general predictable target strategies as in the single-player case investigated in~\cite{BankSonerVoss:17}. Moreover, facing the same time horizon, the players' terminal portfolio positions are also restricted to some exogenously predetermined stochastic levels which reveal gradually over time. As a consequence, both agents will choose their dynamic trading strategies from a suitable set of adapted stochastic processes rather than opting for static strategies from a set of deterministic functions as in the papers cited above (except for the numerical study in \cite{CarmonaYang:09}). 

Other recent work on both finite-player as well as infinite-player mean field price impact games with Almgren-Chriss type price impact include, e.g., \citet{CardaliaguetLehalle:18}, \citet{HuangJaimungalNourian:19}, \citet{CasgrainJaimungal:18, CasgrainJaimungal:19}, \citet{FuGraeweHorstPopier:20}, \citet{FuHorst:20}, \citet{EvangelistaThamsten:20}, and \citet{DrapeauSchied:20}, where finitely and infinitely many agents pursue optimal liquidation of their initial positions and interact through common aggregated permanent and temporary price impact. Price impact games of liquidating agents in a market model with transient price impact are analyzed, e.g., in \citet{LuoSchied:20}, \citet{SchiedZhang:19}, \citet{SchiedStrehleZhang:17}, \citet{Strehle:17}; and very recently in \citet{FuHorstXia:20} and \citet{NeumanVoss:21}. However, these works are all \emph{portfolio liquidation games} where the agents steer their initial portfolio positions towards zero (with strict liquidation constraints enforced in~\cite{FuGraeweHorstPopier:20, FuHorst:20, EvangelistaThamsten:20, FuHorstXia:20}). In particular, the agents neither track any individual stochastic running trading targets nor do they aim for reaching an individual random terminal target position. In contrast, as mentioned above, our present study formulates and solves a two-player price impact \emph{portfolio tracking game} with random terminal state constraints between two heterogeneous agents who have their own individual trading targets. 

Our main result is an \emph{explicit} description of a unique open-loop Nash equilibrium within the class of progressively measurable strategies to our two-player stochastic differential game, where both agents track their own target strategies as in~\cite{BankSonerVoss:17} and interact through temporary and permanent price impact as in~\cite{SchiedZhang:17} and \cite{CarlinLoboViswanathan:07}. Mathematically, we solve a linear quadratic stochastic differential game with random terminal state constraints. Inspired by the analysis in~\cite{BankSonerVoss:17}, we follow a probabilistic and convex-analytic approach in the style of Pontryagin's stochastic maximum principle. This also allows us to consider general predictable strategies as the agents' tracking targets and not necessarily Markovian or continuous diffusion-type processes. We prove uniqueness of the Nash equilibrium and derive its characterization, which takes the form of a four-dimensional coupled system of linear forward-backward stochastic differential equations (FBSDEs). Due to the stochastic terminal state constraints the FBSDE system has singular terminal conditions. As a consequence, explicitly computing a solution to the constrained stochastic differential game is a nontrivial task. The manuscript shows how this can be achieved. Solving the singular FBSDE system provides us with the agents' optimal trading strategies in equilibrium in closed-form and unveils a rich phenomenology for their optimal behaviour. 

In fact, it turns out that in equilibrium, similar to the single-player solution presented in~\cite{BankSonerVoss:17}, both agents anticipate their individual running target portfolio by gradually trading in the direction of a weighted average of expected future target positions of the target strategy. However, being aware of the competitor's tracking goals, each agent also assesses a weighted average of the expected future positions of the opponent's target strategy and chooses to trade accordingly. Interestingly, it arises that the agents' \emph{trading directions} with respect to the adversary's target strategy are not invariant but depend on the relation between temporary and permanent price impact. Conceptually, our explicit results extend the analysis carried out by~\citet{SchiedSchoeneborn:09}. Therein, the authors identify two distinct types of illiquid markets: A \emph{plastic} market where the price impact is predominantly permanent, and an \emph{elastic} market where the major part of incurred price impact is temporary. Their model predicts that a competitor who is conscious of the other agent's liquidation intention engages in predatory trading in a plastic market (in the sense that the competitor partly trades in the \emph{same} direction as her opponent), while she tends to cooperate and provides liquidity in an elastic market (in the sense that she trades in the \emph{opposite} direction of her opponent's trading); cf.~also the detailed discussion in \citet{SchiedSchoeneborn:09}. Our closed-form Nash equilibrium solution of our more general price impact tracking game corroborates this. The novelty of our contribution comes from the fact that \emph{both} predation by simultaneously trading in the same direction as the opponent, as well as cooperation by trading in the opposite direction can occur in a \emph{coexisting} manner; depending on whether the market is plastic or elastic. As a consequence, different behavioral paradigms can emerge as optimal in our Nash equilibrium; see the illustrations in Section~\ref{sec:illustrations}.

The remainder of the paper is organized as follows. In Section~\ref{sec:problem} we introduce our two-player stochastic differential price impact game by extending the framework of~\citet{CarlinLoboViswanathan:07} and~\citet{SchiedZhang:17} to a stochastic tracking problem of general predictable target strategies and random terminal state constraints. Our main result, an explicit description of a unique open-loop Nash equilibrium of the game is presented in Section~\ref{sec:main}. Section~\ref{sec:illustrations} contains some illustrations and discusses the qualitative behaviour of the two players' optimal strategies in equilibrium.

\medskip

\emph{Notation:} Throughout this manuscript we use superscripts for enumerating purposes as, e.g., in $X^1$, $X^2$, $\alpha^1$, $\alpha^2$, or other quantities like $\xi^1$, $\xi^2$ etc., to mark all objects which are associated with player 1 and player 2, respectively; or, to itemize objects as $w^1$, $w^2$, $w^3$ etc. In particular, $X^2$, $\alpha^2$, $\xi^2$ is not to be confused with quadratic powers, which will be explicitly denoted with brackets like $(\alpha)^2$, or, if necessary, as $(\alpha^2)^2$.

\section{Problem formulation} \label{sec:problem}

Let $T >0$ denote a finite deterministic time horizon and fix a
filtered probability space $(\Omega,\cF,(\cF_t)_{0 \leq t \leq T},\P)$
satisfying the usual conditions of right continuity and completeness.
We consider two agents (preferred pronouns she/her/hers and he/him/his, respectively) who are trading in a financial market
consisting of one risky asset, e.g., a stock. The number of shares
agent 1 and agent 2 are holding at time $t \in [0,T]$ are defined,
respectively, as
\begin{equation}
  X^1_t \set x^1 + \int_0^t \alpha^1_s ds \qquad \text{and}  \qquad X^2_t \set
  x^2 + \int_0^t  \alpha^2_s ds
\end{equation}
with initial positions $x^1, x^2 \in \RR$. The real-valued stochastic
processes $(\alpha^1_t)_{0 \leq t \leq T}$ and
$(\alpha^2_t)_{0 \leq t \leq T}$ represent the turnover rate at which
each agent trades in the risky asset and belong to the general class of stochastic processes
\begin{equation} \label{def:setA}
  \cA \set \left\{ \alpha : \alpha \text{ progressively measurable s.t. }
    \mathbb{E} \left[ \int_0^T (\alpha_t)^2 dt \right] < \infty \right\}.
\end{equation}
We adopt the framework from~\citet{CarlinLoboViswanathan:07}
and~\citet{SchiedZhang:17} and suppose that the agents' trading incurs
linear temporary and permanent price impact \`a la~\citet{AlmgChr:01}
in the sense that trades in the risky asset are executed at prices
\begin{equation} \label{def:exPrice}
  S_ t \set P_t + \lambda (\alpha^1_t + \alpha^2_t)+ \gamma ((X^1_t -
  x^1) + (X^2_t - x^2)) \quad (0
  \leq t \leq T)
\end{equation}
with some unaffected price process
$P_\cdot= P_0 + \sqrt{\sigma} W_\cdot$ following a Brownian motion
$(W_t)_{0 \leq t \leq T}$ with respect to the underlying filtration with variance $\sigma > 0$. The trading of
both agents in the risky asset consumes available liquidity and
instantaneously affects the execution price in~\eqref{def:exPrice} in
an adverse manner through temporary price impact $\lambda > 0$. In
addition, the agents' total accumulated trading activity also leaves a
trace in the execution price which is captured by the permanent price
impact parameter $\gamma >0$.

Similar to the single-agent setup in~\citet{BankSonerVoss:17} we
assume that agent~1 and agent~2 are trading in this illiquid risky
asset because each agent seeks to track their own exogenously given
target strategy $(\xi_t^1)_{0 \leq t \leq T}$ and
$(\xi_t^2)_{0 \leq t \leq T}$, respectively. Both processes $\xi^1$
and $\xi^2$ are supposed to be real-valued predictable processes in
$L^2(\PP \otimes dt)$ and can be thought of, for instance, as hedging
strategies adopted from a frictionless market. Moreover, the agents
are also required to reach a predetermined terminal portfolio target
position $\Xi^1_T$ and $\Xi^2_T$ in $L^2(\PP,\cF_T)$ at time
$T$. Mathematically, we can formalize their objectives as follows: For
a given strategy $(\alpha^2_t)_{0 \leq t \leq T}$ of her competitor
agent~2, agent~1 aims to choose her trading rate
$(\alpha^1_t)_{0 \leq t \leq T}$ in order to minimize the cost
functional
\begin{equation}
  \label{def:objective1} \begin{aligned}
    J^1(\alpha^1;\alpha^2) \set
    & \; \E \left[ \frac{1}{2} \sigma \int_0^T
      (X^1_t
      - \xi^1_t)^2 dt \right. \\
    & \left. \hspace{14pt} + \frac{1}{2} \lambda \int_0^T
      \alpha^1_t \left( \alpha^1_t + \alpha^2_t \right) dt + \frac{1}{2} \gamma
      \int_0^T \alpha^1_t \left( X^2_t - x^2 \right) dt \right],
  \end{aligned}
\end{equation}
whereas agent~2 wishes to minimize
\begin{equation} \label{def:objective2} \begin{aligned}
    J^2(\alpha^2;\alpha^1) \set
    & \; \E \left[ \frac{1}{2} \sigma \int_0^T
      (X^2_t
      - \xi^2_t)^2 dt \right. \\
    & \left. \hspace{14pt} +
    \frac{1}{2} \lambda \int_0^T
    \alpha^2_t \left( \alpha^1_t + \alpha^2_t \right) dt + \frac{1}{2} \gamma
    \int_0^T \alpha^2_t \left( X^1_t - x^1 \right) dt \right]
\end{aligned}
\end{equation}
his trading rate $(\alpha^2_t)_{0 \leq t \leq T}$ in response to a
given strategy $(\alpha^1_t)_{0 \leq t \leq T}$ of his opponent
agent~1. As in the single-agent problem in~\citet{BankSonerVoss:17},
the first term in~\eqref{def:objective1} and~\eqref{def:objective2}
reflects the agents' running after their individual target strategies
$\xi^1$ and $\xi^2$, respectively, through minimizing the
corresponding square deviation from their respective portfolio
positions~$X^1$ and~$X^2$. The common weight parameter $\sigma$
measures price fluctuations of the underlying unaffected price
process. The second and third terms in~\eqref{def:objective1}
and~\eqref{def:objective2} take into account the additional incurred
linear quadratic illiquidity costs which are induced by temporary and
permanent price impact while both agents are trading in the risky
asset as stipulated in~\eqref{def:exPrice} (see also~\citet{CarlinLoboViswanathan:07} and~\citet{SchiedZhang:17}). 
Note, however, that due to each agent's
individual terminal state constraint $X^i_T = \Xi^i_T$ $\PP$-a.s. (for
$i=1,2$) only the competitor's accrued permanent price impact feeds
into their respective cost functional. Indeed, integration by parts
yields that the $i$-th agent's permanent impact from their own trading
always creates the same costs
$\gamma (X^i_T - x^i)^2=\gamma (\Xi^i_T - x^i)^2$ independent of their
chosen trading rate and therefore can be neglected in their own
objective functional. We obtain following individual optimal
stochastic control problems for agent~1 and agent~2, namely,
\begin{equation} \label{eq:optProbAgent1}
  J^1(\alpha^1;\alpha^2) \rightarrow \min_{\alpha^1 \in \cA^1}
\end{equation}
for any fixed strategy $\alpha^2 \in \cA^2$, and
\begin{equation} \label{eq:optProbAgent2}
  J^2(\alpha^2;\alpha^1)  \rightarrow \min_{\alpha^2 \in \cA^2}, 
\end{equation}
for any fixed strategy $\alpha^1 \in \cA^1$, where
$\cA^{i}$, $i=1,2$, is the set of admissible constrained policies
defined as
\begin{equation} \label{def:admissibilitySet}
  \cA^{i} \set \left\{ \alpha^i : \alpha^i \in \cA \text{ satisfying
    } X^i_T = x^i + \int_0^T \alpha^i_t dt = 
    \Xi^i_T \; \PP\text{-a.s.} \right\}.
\end{equation}
Similar to~\citet{BankSonerVoss:17} we further assume that the target
positions $\Xi^1_T, \Xi^2_T \in L^2(\PP,\cF_T)$ satisfy
\begin{equation} \label{eq:assumption}
  \EE \left[ \int_0^T
      \frac{1}{T-s} d\langle M^+ \rangle_s \right] < \infty \quad
    \text{and} \quad \EE \left[ \int_0^T
      \frac{1}{T-s} d\langle M^- \rangle_s \right] < \infty,
\end{equation}
where $M^+_t \set \EE[\Xi^1_T + \Xi^2_T \vert \cF_t]$ and $M^-_t \set
\EE[\Xi^1_T - \Xi^2_T \vert \cF_t]$ for $0 \leq t \leq T$. 

\begin{Remark}
  \begin{enumerate}
  \item As in \citet{CarlinLoboViswanathan:07}
    and~\citet{SchiedZhang:17} the agent's individual optimization
    problems in~\eqref{eq:optProbAgent1} and~\eqref{eq:optProbAgent2}
    are intertwined through common aggregated temporary and permanent
    price impact affecting their performance functionals $J^1$ and
    $J^2$ in~\eqref{def:objective1} and~\eqref{def:objective2} (in
    contrast to, e.g,~\citet{HuangJaimungalNourian:19},
    \citet{CasgrainJaimungal:18, CasgrainJaimungal:19} or
    \citet{EkrenNadtochiy:19} where agents only interact through
    permanent or temporary price impact, respectively). One can think
    of both players as \emph{strategic agents} who compete for
    liquidity while concurrently trading in a single illiquid risky
    asset to meet their tracking objectives for the purpose of, e.g.,
    hedging fluctuations of random endowments. Note that both agents
    are fully aware of the opponent's trading targets~$\xi^i$
    and~$\Xi_T^i$ ($i=1,2$), as well as the jointly caused price
    impact on the execution prices in~\eqref{def:exPrice}. That is,
    our game is one of complete information as in the related studies
    in~\citet{BrunnermeierPedersen:05},~\citet{CarlinLoboViswanathan:07},
    \citet{SchiedSchoeneborn:09},~\citet{CarmonaYang:09},
    and~\citet{SchiedZhang:17}.
  \item For further motivation for the tracking cost functionals
    in~\eqref{def:objective1} and~\eqref{def:objective2} we refer to
    the single-player optimization problems studied, e.g., in~\citet{RogerSin:10}, \citet{NaujWes:11},~\citet{HorstNaujokat:14},~\citet{AlmgLi:16},~\citet{BankSonerVoss:17}, 
    and~\citet{CaiRosenbaumTankov:17}. Observe that the square
    tracking error also incorporates a risk aversion on each player's
    inventory. In this regard, both agents are homogeneous in their
    inventory risk.
    \item Note that the coefficients $\sigma, \lambda, \gamma > 0$ in the cost functionals in~\eqref{def:objective1} and~\eqref{def:objective2} are constants. This is an important assumption for obtaining a closed-from solution for the stochastic differential game, which is our primary focus of interest. In fact, the only sources of randomness in the game are the target strategies $(\xi^1_t)_{0 \leq t \leq T}$, $(\xi^2_t)_{0 \leq t \leq T}$ and the random terminal conditions $\Xi^1_T$, $\Xi^2_T$, which will force the agents' optimal policies to be random processes as well.
  \item Analog to the study in~\citet{BankSonerVoss:17} the assumption
    in~\eqref{eq:assumption} will ensure that $\cA^i \neq \varnothing$
    for $i=1,2$ (cf. also the proof of Theorem~\ref{thm:main} in
    Section~\ref{sec:main} below). In fact, for given random
    variables $\Xi^i_T \in L^2(\PP,\cF_T)$ only known at time $T$ the
    terminal state constraint $X^i_T = \Xi^i_T$ $\PP$-a.s. ($i=1,2$)
    is quite demanding. Thus, loosely speaking, the condition
    in~\eqref{eq:assumption} requires that the speed at which
    information on the random ultimate target positions $\Xi^1_T$,
    $\Xi^2_T$ is revealed as $t \uparrow T$ is sufficiently fast.
  \end{enumerate}
\end{Remark}

Our goal is to compute a Nash equilibrium in which both agents solve
their minimization problems in~\eqref{eq:optProbAgent1} and
\eqref{eq:optProbAgent2} simultaneously, given the strategy of their
competitor, in the following sense:

\begin{Definition} \label{def:Nash} A pair of admissible strategies
  $ (\hat{\alpha}^1,\hat{\alpha}^2) \in \cA^1 \times \cA^2$ is called
  an \emph{open-loop Nash equilibrium} if for all admissible
  strategies $\alpha^1 \in \cA^1$ and $\alpha^2 \in \cA^2$ it holds
  that
  \begin{equation*}
    J^1 (\hat{\alpha}^1;\hat{\alpha}^2) \leq J^1 (\alpha^1;\hat{\alpha}^2)
    \quad 
    \text{and} \quad  J^2 (\hat{\alpha}^2;\hat{\alpha}^1) \leq J^2
    (\alpha^2;\hat{\alpha}^1).
  \end{equation*}
\end{Definition}

In other words, in a Nash equilibrium neither player has an incentive to
deviate from the chosen strategy.

\begin{Remark} \label{rem:problem}
  In the special case of optimally liquidating the agents'
    initial risky asset holdings $x^1, x^2 \in \RR$ without tracking
    exogenously given target strategies, i.e.,
    $\xi^1 \equiv \xi^2 \equiv 0$, and with non-random terminal target
    positions $\Xi^1_T = \Xi^2_T = 0$ $\PP$-almost surely, the above
    formulated two-player (deterministic) differential game is solved
    in~\citet{CarlinLoboViswanathan:07} setting $\sigma = 0$ in the
    performance functionals in~\eqref{def:objective1}
    and~\eqref{def:objective2}; and in~\citet{SchiedZhang:17} allowing
    for $\sigma > 0$ instead. In both studies, the authors obtain a
    unique open-loop Nash equilibrium in the sense of
    Definition~\ref{def:Nash} in closed form within the class of
    deterministic strategies.
\end{Remark}

\section{Main result} \label{sec:main}

Our main result is an explicit description of a unique open-loop Nash equilibrium in the sense of Definition~\ref{def:Nash} of the two-player stochastic differential game formulated in Section~\ref{sec:problem}. Inspired by~\citet{BankSonerVoss:17} we will use tools from convex analysis and simple calculus of variations arguments to derive the equilibrium strategies. 

First, a strict convexity property of each players' objective in~\eqref{def:objective1} and~\eqref{def:objective2} is established in the following

\begin{Lemma} \label{lem:convex}
For every $\alpha^2 \in \mathcal{A}^2$ fixed, the functional $\alpha^1 \mapsto J^1(\alpha^1;\alpha^2)$ in \eqref{def:objective1} is strictly convex in $\alpha^1 \in \mathcal{A}^1$. Similarly, for every $\alpha^1 \in \mathcal{A}^1$ fixed, the functional $\alpha^2 \mapsto J^2(\alpha^2;\alpha^1)$ in \eqref{def:objective2} is strictly convex in $\alpha^2 \in \mathcal{A}^2$.
\end{Lemma}

\begin{proof}
We only show strict convexity of the first agent's objective in~\eqref{def:objective1}. The reasoning for the second agent's objective in~\eqref{def:objective2} follows analogously. 
To this end, let $\alpha^2 \in \mathcal{A}^2$ be fixed. Consider $\alpha^1,\tilde{\alpha}^1 \in \mathcal{A}^1$ such that $\alpha^1 \neq \tilde{\alpha}^1$ $d\P \otimes dt\textrm{-a.e. on } \Omega \times[0,T]$ and denote by $X^1, \tilde{X}^1$ the corresponding share holdings. For every $\varepsilon \in (0,1)$ it holds that $\varepsilon \alpha^1 + (1-\varepsilon) \tilde{\alpha}^1 \in \mathcal{A}^1$ with share holdings $X^{\varepsilon \alpha^1 + (1-\varepsilon) \tilde{\alpha}^1} =
\varepsilon X^{1} + (1-\varepsilon) \tilde{X}^1$. We have to show that
\begin{equation*}
\varepsilon J^1(\alpha^1;\alpha^2) + (1-\varepsilon) J^1(\tilde{\alpha}^1;\alpha^2) - J^1(\varepsilon \alpha^1 + (1-\varepsilon) \tilde{\alpha}^1; \alpha^2) > 0.
\end{equation*}
In fact, a straightforward computation reveals that
\begin{equation*}
  \begin{aligned}
    & \varepsilon J^1(\alpha^1;\alpha^2) + (1-\varepsilon) J^1(\tilde{\alpha}^1;\alpha^2) - J^1(\varepsilon \alpha^1 + (1-\varepsilon) \tilde{\alpha}^1; \alpha^2) \\
    & = \frac{1}{2} \varepsilon (1-\varepsilon) \mathbb{E} \left[ \int_0^T \left(
    \sigma (X^1_t- \tilde{X}^1_t)^2+\lambda (\alpha^1_t - \tilde{\alpha}^1_t)^2 \right) dt \right] >0
  \end{aligned}
\end{equation*}
because $\alpha^1 \neq \tilde{\alpha}^1$ $d\P \otimes ds\textrm{-a.e. on } \Omega \times[0,T]$.
\end{proof}

As an important consequence we obtain
\begin{Lemma} \label{lem:uniqueNE}
  There exists at most one Nash equilibrium in the sense of
  Definition~\ref{def:Nash}.
\end{Lemma}

\begin{proof}
  We adapt the argument from~\citet[Lemma 4.1]{SchiedZhang:17} (see also~\citet[Proposition 4.8]{SchiedStrehleZhang:17}) to our stochastic differential game and prove the claim by contradiction. Specifically, assume that there exist two distinct Nash equilibria $(\hat{\alpha}^1,\hat{\alpha}^2)$ and $(\tilde{\alpha}^1,\tilde{\alpha}^2)$ in $\cA^1 \times \cA^2$, i.e.,
    \begin{equation} \label{proof:unique:NEs}
    \begin{aligned}
    J^1 (\hat{\alpha}^1;\hat{\alpha}^2) \leq J^1 (\alpha^1;\hat{\alpha}^2)
    \quad 
    \text{and} \quad  J^2 (\hat{\alpha}^2;\hat{\alpha}^1) \leq J^2
    (\alpha^2;\hat{\alpha}^1), \\
    J^1 (\tilde{\alpha}^1;\tilde{\alpha}^2) \leq J^1 (\alpha^1;\tilde{\alpha}^2)
    \quad 
    \text{and} \quad  J^2 (\tilde{\alpha}^2;\tilde{\alpha}^1) \leq J^2
    (\alpha^2;\tilde{\alpha}^1),
    \end{aligned}
  \end{equation}
  for all admissible strategies $\alpha^1 \in \cA^1$ and $\alpha^2 \in \cA^2$. Then we can define for all $\varepsilon \in [0,1]$ the function
  \begin{equation} \label{proof:unique:deff}
  \begin{aligned}
      f(\varepsilon) \triangleq & \, J^1(\varepsilon \tilde{\alpha}^1 + (1-\varepsilon) \hat{\alpha}^1;\hat{\alpha}^2) + J^2(\varepsilon \tilde{\alpha}^2 + (1-\varepsilon) \hat{\alpha}^2;\hat{\alpha}^1) \\
      & \, + J^1((1-\varepsilon) \tilde{\alpha}^1 + \varepsilon \hat{\alpha}^1 ;\tilde{\alpha}^2) + J^2((1-\varepsilon) \tilde{\alpha}^2 + \varepsilon \hat{\alpha}^2 ;\tilde{\alpha}^1) .
  \end{aligned}
  \end{equation}
  Note that due to Lemma~\ref{lem:convex} and the assumption that the two Nash equilibria $(\hat{\alpha}^1,\hat{\alpha}^2)$ and $(\tilde{\alpha}^1,\tilde{\alpha}^2)$ are distinct, the function $f(\varepsilon)$ is strictly convex in $\varepsilon$ on $[0,1]$. Moreover, in light of~\eqref{proof:unique:NEs} it has a unique minimum in $\varepsilon = 0$. It follows that
  \begin{equation} \label{proof:unique:fpos}
      \lim_{\varepsilon \downarrow 0} \frac{f(\varepsilon)-f(0)}{\varepsilon} = \frac{d}{d\varepsilon} f(\varepsilon) \Big\vert_{\varepsilon = 0+} \geq 0.
  \end{equation}
  Next, denoting the corresponding share holdings of $\hat{\alpha}^1$ and $\tilde{\alpha}^1$ with $\hat{X}^1$ and $\tilde{X}^1$, respectively, and noting that $X^{\varepsilon \tilde{\alpha}^1 + (1-\varepsilon) \hat{\alpha}^1} = \varepsilon \tilde{X}^1 + (1-\varepsilon) \hat{X}^1$, we can compute
    \begin{equation*} 
  \begin{aligned}
      & \frac{d}{d\varepsilon} J^1(\varepsilon \tilde{\alpha}^1 + (1-\varepsilon) \hat{\alpha}^1;\hat{\alpha}^2) \Big\vert_{\varepsilon = 0+} \\
      & = \mathbb{E} \left[ \sigma\! \int_0^T (\hat{X}^1_t - \xi^1_t) (\tilde{X}^1_t - \hat{X}^1_t) dt + \int_0^T (\tilde{\alpha}^1_t-\hat{\alpha}^1_t) \left( \frac{1}{2} \lambda (2\hat{\alpha}^1_t + \hat{\alpha}^2_t) + \frac{1}{2} \gamma (\hat{X}^2_t - x^2) \right)dt \right],
  \end{aligned}
  \end{equation*}
  as well as the derivatives of the remaining three terms in~\eqref{proof:unique:deff} in a very similar manner in order to ultimately obtain
  \begin{equation*} 
  \begin{aligned}
      & \frac{d}{d\varepsilon} f(\varepsilon) \Big\vert_{\varepsilon = 0+} \\
      & = - \sigma \mathbb{E} \left[ \int_0^T \left( (\tilde{X}^1_t - \hat{X}^1_t)^2 + (\tilde{X}^2_t - \hat{X}^2_t)^2 \right) dt \right] \\
      & \quad + \frac{1}{2} \gamma \mathbb{E}\left[ \int_0^T (\tilde{\alpha}^1_t - \hat{\alpha}^1_t) (\hat{X}^2_t - \tilde{X}^2_t) dt \right] + \frac{1}{2} \gamma \mathbb{E}\left[ \int_0^T (\tilde{\alpha}^2_t - \hat{\alpha}^2_t) (\hat{X}^1_t - \tilde{X}^1_t) dt \right] \\
      & \quad - \lambda \mathbb{E} \left[ \int_0^T \left( (\tilde{\alpha}^1_t - \hat{\alpha}^1_t) + (\tilde{\alpha}^2_t - \hat{\alpha}^2_t) \right)^2 dt \right],
  \end{aligned}
  \end{equation*}
  where $\hat{X}^2$ and $\tilde{X}^2$ denote the share holdings of $\hat{\alpha}^2$ and $\tilde{\alpha}^2$, respectively. Observing that integration by parts yields
  \begin{equation*}
      \int_0^T (\tilde{\alpha}^1_t - \hat{\alpha}^1_t) (\hat{X}^2_t - \tilde{X}^2_t) dt = - \int_0^T (\tilde{\alpha}^2_t - \hat{\alpha}^2_t) (\hat{X}^1_t - \tilde{X}^1_t) dt
  \end{equation*}
  because $\tilde{X}^i_0 = \hat{X}^i_0 = x^i$ and $\hat{X}^i_T = \tilde{X}^i_T = \Xi^i_T$ for both $i \in \{1,2\}$, we obtain
  \begin{equation*}
    \begin{aligned}
  \frac{d}{d\varepsilon} f(\varepsilon) \Big\vert_{\varepsilon = 0+} = & \, - \sigma \mathbb{E} \left[ \int_0^T \left( (\tilde{X}^1_t - \hat{X}^1_t)^2 + (\tilde{X}^2_t - \hat{X}^2_t)^2 \right) dt \right] \\
  & \, - \lambda \mathbb{E} \left[ \int_0^T \left( (\tilde{\alpha}^1_t - \hat{\alpha}^1_t) + (\tilde{\alpha}^2_t - \hat{\alpha}^2_t) \right)^2 dt \right]
  \end{aligned}
  \end{equation*}
  which is strictly negative because the two Nash equilibria $(\hat{\alpha}^1,\hat{\alpha}^2)$ and $(\tilde{\alpha}^1,\tilde{\alpha}^2)$ are distinct. But this contradicts~\eqref{proof:unique:fpos}.
\end{proof}

Next, for any arbitrary but fixed controls $\tilde{\alpha}^2 \in \cA^2$ and $\tilde{\alpha}^1 \in \cA^1$, we can introduce the G\^ateaux derivatives of the mappings $\alpha^1 \mapsto J^1(\alpha^1;\tilde{\alpha}^2)$ at $\alpha^1 \in \cA^1$ and $\alpha^2 \mapsto J^2(\alpha^2;\tilde{\alpha}^1)$ at $\alpha^2 \in \cA^2$, respectively, in any directions
$\beta^1, \beta^2 \in \cA^0 \set \{ \beta : \beta \in \cA \text{
  satisfying } \int_0^T \beta_t dt = 0 \; \PP\text{-a.s.}\}$, namely,
\begin{align*}
  \langle \nabla J^1(\alpha^1; \tilde{\alpha}^2), \beta^1 \rangle \set
  & \; \lim_{\varepsilon \rightarrow 0} \frac{J^1(\alpha^1 +\varepsilon
    \beta^1;\tilde{\alpha}^2)-J^1(\alpha^1,\tilde{\alpha}^2)}{\varepsilon},\\
  \langle \nabla J^2(\alpha^2;\tilde{\alpha}^1), \beta^2 \rangle \set
  & \; \lim_{\varepsilon \rightarrow 0}
    \frac{J^2(\alpha^2 +\varepsilon 
    \beta^2; \tilde{\alpha}^1 )-J^2(\alpha^2;\tilde{\alpha}^1)}{\varepsilon}.  
\end{align*}
They allow for following explicit expressions presented in

\begin{Lemma} \label{lem:gateaux}
Let $\tilde{\alpha}^2 \in \cA^2$ be fixed with corresponding share holdings $\tilde{X}^2$. Then for all $\alpha^1 \in \cA^1$ we have
  \begin{align}
    & \langle \nabla J^1(\alpha^1; \tilde{\alpha}^2), \beta^1 \rangle
      \nonumber \\
    & = \EE \left[ \int_0^T
      \beta^1_s \left( \lambda \alpha^1_s + \frac{\lambda}{2} 
      \tilde{\alpha}^2_s + \frac{\gamma}{2} (\tilde{X}^2_s - x^2) + \int_s^T (X^1_t - \xi^1_t)
      \sigma dt \right) ds \right] \label{eq:gateaux1}
    \end{align}
for any $\beta^1 \in \cA^0$. Similarly, let $\tilde{\alpha}^1 \in \cA^1$ be fixed with corresponding share holdings $\tilde{X}^1$. Then for all $\alpha^2 \in \cA^2$ we have    
    \begin{align}
    & \langle \nabla J^2(\alpha^2; \tilde{\alpha}^1), \beta^2 \rangle \nonumber \\
    & = \EE \left[ \int_0^T
      \beta^2_s \left( \lambda \alpha^2_s + \frac{\lambda}{2} 
      \tilde{\alpha}^1_s + \frac{\gamma}{2} (\tilde{X}^1_s - x^1) + \int_s^T (X^2_t - \xi^2_t)
      \sigma dt \right) ds \right]  \label{eq:gateaux2}
  \end{align}
  for any $\beta^2 \in \cA^0$.
\end{Lemma}

\begin{proof}
  We only compute the G\^ateaux derivative in~\eqref{eq:gateaux1}. The
  same computations apply for~\eqref{eq:gateaux2}. Fix $\tilde{\alpha}^2 \in \cA^2$ with share holdings $\tilde{X}^2$ and let
  $\alpha^1 \in \cA^1, \beta^1 \in \cA^0$ as well as $\varepsilon > 0$. Note
  that $\alpha^1 +\varepsilon \beta^1 \in \mathcal{A}^1$ with share holdings $X^{\alpha^1 +\varepsilon \beta^1} = X^1 +
  \varepsilon \int_0^\cdot \beta^1_s ds$. Moreover, since
  \begin{align*}
    &  J^1(\alpha^1 +\varepsilon \beta^1; \tilde{\alpha}^2)-J^1(\alpha^1;\tilde{\alpha}^2)
    \\
    & =  \, \varepsilon \EE \left[ \int_0^T \left( \frac{\lambda}{2} \beta^1_t
      (2 \alpha^1_t + \tilde{\alpha}^2_t) + 
      \left(\int_0^t \beta^1_s ds \right) (X^1_t - \xi^1_t) \sigma
      +\frac{\gamma}{2} \beta^1_t (\tilde{X}^2_t -x^2) \right) dt \right] \\
    & \quad + \frac{1}{2} \epsilon^2 \EE \left[ \int_0^T \left( \lambda 
      (\beta^1_t)^2 + \left(\int_0^t \beta^1_sds \right)^2 \sigma
      \right) dt \right],
  \end{align*}
  we obtain the desired result in~\eqref{eq:gateaux1} after applying
  Fubini's theorem.
\end{proof}

Having at hand the explicit expressions in~\eqref{eq:gateaux1}
and~\eqref{eq:gateaux2} we can now derive a sufficient and necessary first
order condition for the Nash equilibrium in terms of a system of coupled forward-backward stochastic differential equations (FBSDE).

\begin{Lemma} \label{lem:FOC}
  A pair of controls
  $(\hat{\alpha}^1,\hat{\alpha}^2) \in \cA^1 \times \cA^2$ is a Nash
  equilibrium in the sense of Definition~\ref{def:Nash} if and only if
  $(\hat{X}^1, \hat{X}^2, \hat{\alpha}^1,\hat{\alpha}^2)$ solve
  following coupled forward backward SDE system
  \begin{equation} \label{eq:FBSDE}
    \left\{
    \begin{aligned}
      dX^1_t = & \; \alpha^1_t dt, \qquad X^1_0 = x^1, \\
      dX^2_t = & \; \alpha^2_t dt, \qquad X^2_0 = x^2, \\
      d\alpha^1_t = & \; \frac{\sigma}{\lambda} (X^1_t - \xi^1_t) dt -
      \frac{\gamma}{2\lambda} \alpha^2_t dt - \frac{1}{2} d\alpha^2_t +
      dM^1_t, \qquad X^1_T =
      \Xi^1_T,\\
      d\alpha^2_t = & \; \frac{\sigma}{\lambda} (X^2_t - \xi^2_t) dt -
      \frac{\gamma}{2\lambda} \alpha^1_t dt - \frac{1}{2} d\alpha^1_t +
      dM^2_t, \qquad X^2_T = \Xi^2_T,
    \end{aligned}
    \right.
  \end{equation}
  for two suitable square integrable martingales
  $(M^1_t)_{0 \leq t < T}$ and $(M^2_t)_{0 \leq t < T}$.
\end{Lemma}

\begin{proof}
  \emph{Sufficiency:} Assume first that
  $(\hat{X}^1, \hat{X}^2,\hat{\alpha}^1,\hat{\alpha}^2, M^1, M^2)$
  with $(\hat{\alpha}^1,\hat{\alpha}^2) \in \cA^1 \times \cA^2$ solves
  the FBSDE system in~\eqref{eq:FBSDE}.  We have to show that
  $\hat{\alpha}^1$ minimizes
  $\alpha^1 \mapsto J^1(\alpha^1;\hat{\alpha}^2)$ over $\cA^1$, and,
  vice versa, that $\hat{\alpha}^2$ minimizes
  $\alpha^2 \mapsto J^2(\alpha^2;\hat{\alpha}^1)$ over $\cA^2$. Since
  we are minimizing strictly convex functionals due to Lemma~\ref{lem:convex}, a sufficient
  condition for the optimality of $\hat{\alpha}^1$ and
  $\hat{\alpha}^2$, respectively, is given by
  \begin{equation} \label{eq:proofFBSDE:1} 
  \langle \nabla
    J^1(\hat{\alpha}^1; \hat{\alpha}^2), \beta^1 \rangle = 0 \text{
      for all } \beta^1 \in \cA^0
  \end{equation}
  and
  \begin{equation} \label{eq:proofFBSDE:2} \langle \nabla
    J^2(\hat{\alpha}^2; \hat{\alpha}^1), \beta^2 \rangle = 0 \text{
      for all } \beta^2 \in \cA^0;
  \end{equation}
  cf., e.g., \citet{EkelTem:99}. We start
  with the proof of~\eqref{eq:proofFBSDE:1}. By assumption we have the
  representation
  \begin{align*}
    \hat{\alpha}^1_t =
    & \; \hat{\alpha}^1_0 + \frac{\sigma}{\lambda} \int_0^t
      (\hat{X}^1_s - \xi^1_s) ds- \frac{\gamma}{2\lambda} \int_0^t
      \hat{\alpha}^2_s ds \\
    & -\frac{1}{2} (\hat{\alpha}^2_t -
      \hat{\alpha}^2_0) + M^1_t-M^1_0 \quad d\PP \otimes dt 
      \text{-a.e. on } \Omega \times [0,T)
  \end{align*}
  for some square integrable martingale $(M^1_t)_{0 \leq t < T}$.
  Moreover, since
  $\hat{\alpha}^1, \hat{\alpha}^2, \xi^1 \in L^2(\PP \otimes dt)$ it
  follows that $\EE[\int_0^T (M_s^1)^2 ds] < \infty$. Next,
  introducing the square integrable martingale
  \begin{equation*}
    N_s \set \EE \left[ \int_0^T (\hat{X}^1_t - \xi^1_t) \sigma
      dt \, \bigg\vert \, \cF_s \right] \quad (0 \leq s \leq T) 
  \end{equation*}
  and plugging the above representation of $\hat{\alpha}^1$ in the
  G\^ateaux derivative in~\eqref{eq:gateaux1} we obtain
  \begin{align*}
    & \langle \nabla_1 J^1(\hat{\alpha}^1; \hat{\alpha}^2), \beta^1 \rangle \\
    &= \EE \left[ \int_0^T
      \beta^1_s \left( \lambda \hat{\alpha}^1_s + \frac{\lambda}{2} 
      \hat{\alpha}^2_s + \frac{\gamma}{2} (\hat{X}^2_s - x^2) + \int_s^T
      (\hat{X}^1_t - \xi^1_t) 
      \sigma dt \right) ds \right] \\
    & = \EE \left[ \int_0^T \beta^1_s \left( \lambda \hat{\alpha}^1_0 +
      \frac{\lambda}{2} \hat{\alpha}^2_0 +
      N_T + \lambda M^1_s - \lambda M^1_0 \right) ds  \right] \\
    & = \EE \left[ \left( \lambda \hat{\alpha}^1_0 +
      \frac{\lambda}{2} \hat{\alpha}^2_0 + N_T - \lambda M^1_0 \right) \int_0^T \beta^1_s ds
      \right] + \lambda \EE \left[ \int_0^T 
      \beta^1_s M^1_s ds \right] \\ 
    & = 0 \text{ for all } \beta^1 \in  \cA^0,
  \end{align*}
  where we used the result from~\citet[Lemma 5.3]{BankSonerVoss:17} in
  the last line. Hence, as desired, we obtain that the first order
  optimality condition in~\eqref{eq:proofFBSDE:1} is satisfied by
  $\hat{\alpha}^1 \in \cA^1$. In fact, the same computations apply to
  show that also $\hat{\alpha}^2 \in \cA^2$ is satisfying the first
  order optimality condition in~\eqref{eq:proofFBSDE:2}. Therefore, we
  can conclude that
  $(\hat{\alpha}^1,\hat{\alpha}^2) \in \cA^1 \times \cA^2$ is a Nash
  equilibrium in the sense of Definition~\ref{def:Nash}.

  \emph{Necessity:} Finally, as shown in the proof of
  Theorem~\ref{thm:main} below (which does not use the necessity
  assertion of the present lemma) the pair of controls 
  $(\hat{\alpha}^1, \hat{\alpha}^2) \in \cA^1 \times \cA^2$
  presented in~\eqref{eq:Xsol} below satisfies the coupled forward backward SDE system
  in~\eqref{eq:FBSDE}. Therefore, by uniqueness of the Nash
  equilibrium via Lemma~\ref{lem:uniqueNE} the assertion is indeed also
  necessary.
\end{proof}

We are now ready to state our main result. To do so, it is convenient to
introduce following nonnegative constants
\begin{equation} \label{def:deltaplusminus}
  \delta^+ \set \frac{\gamma^2}{4} + 6 \lambda \sigma, \qquad \delta^- \set
  \frac{\gamma^2}{4} + 2 \lambda \sigma,
\end{equation}
the nonnegative functions
\begin{equation} \label{def:cpcm}
  \begin{aligned}
    c^+_t \set & \; \frac{1}{3} \sqrt{\delta^+}
    \coth(\sqrt{\delta^+}(T-t)/(3\lambda)) + \frac{1}{6} \gamma, \\
    c^-_t \set & \; \sqrt{\delta^-}
    \coth(\sqrt{\delta^-}(T-t)/\lambda) -\frac{1}{2}\gamma 
  \end{aligned} \qquad (0 \leq t \leq T)
\end{equation}
such that $\lim_{t \uparrow T} c_t^{\pm} = +\infty$, as well as the
\emph{weight functions}
\begin{equation}
  \begin{aligned} \label{def:weights}
    w^1_t \set & \; \frac{\sqrt{\delta^+} \,
      e^{\frac{\gamma}{6\lambda}(T-t)}}{3 (c^+_t + c^-_t)
      \sinh(\sqrt{\delta^+}(T-t) /(3\lambda))}, \\
    w^2_t \set &\; \frac{\sqrt{\delta^-} \,
      e^{-\frac{\gamma}{2\lambda}(T-t)}}{(c^+_t + c^-_t)
      \sinh(\sqrt{\delta^-}(T-t)/\lambda)},
    \\
    w^3_t \set & \; \frac{c^+_t}{c^+_t + c^-_t} - w^1_t, \qquad
    w^4_t \set \frac{c^-_t}{c^+_t + c^-_t} - w^2_t, \qquad
    w^5_t \set \frac{c^+_t - c^-_t}{c^+_t + c^-_t}
  \end{aligned} 
\end{equation}
for all $t \in [0,T]$. An explicit description of the unique Nash equilibrium is provided in the following 

\begin{Theorem} \label{thm:main} There exists a unique open-loop Nash
  equilibrium $(\hat{\alpha}^1, \hat{\alpha}^2)$ in
  $\cA^1 \times \cA^2$ in the sense of Definition~\ref{def:Nash}. The
  corresponding equilibrium share holdings
  $\hat{X}^1_\cdot = x^1 + \int_0^\cdot \hat{\alpha}^1_tdt$ of agent~1
  and $\hat{X}^2_\cdot = x^2 + \int_0^\cdot \hat{\alpha}^2_tdt$ of
  agent~2 satisfy the random linear coupled ODE
  \begin{equation} \label{eq:Xsol}
    \begin{aligned}
      \hat{X}^1_0 = & \; x^1, & \quad d\hat{X}^1_t = & \;
      \frac{c^+_t+c^-_t}{2\lambda} \left( \hat{\xi}^1_t - w^5_t
        \hat{X}^2_t - \hat{X}^1_t \right) dt, \\
      \hat{X}^2_0 = & \; x^2, & \quad d\hat{X}^2_t = & \;
      \frac{c^+_t+c^-_t}{2\lambda} \left( \hat{\xi}^2_t - w^5_t
        \hat{X}^1_t - \hat{X}^2_t \right) dt
    \end{aligned}
    \quad (0 \leq t < T),
  \end{equation}
  where, for $0 \leq t \leq T$, we let
  \begin{equation}
    \begin{aligned} \label{def:optsignal1}
      \hat{\xi}^{1}_t \set
      & \; w^1_t \cdot \EE[\Xi^1_T + \Xi^2_T \,\vert\, \cF_t] + w^2_t
      \cdot \EE[\Xi^1_T - \Xi^2_T \,\vert\, \cF_t] \\
      & + w^3_t \cdot \EE\left[ \int_t^T (\xi^1_u + \xi^2_u) \cdot
        K^1(t,u) \,du \,\Big\vert \,\cF_t \right] \\
      & + w^4_t \cdot \EE\left[ \int_t^T (\xi^1_u - \xi^2_u) \cdot
        K^2(t,u) \, du \, \Big\vert \, \cF_t \right]
    \end{aligned}
  \end{equation}
  and
  \begin{equation}
    \begin{aligned} \label{def:optsignal2}
      \hat{\xi}^2_t \set
      & \; w^1_t \cdot \EE[\Xi^2_T + \Xi^1_T \,\vert\, \cF_t] + w^2_t
      \cdot \EE[\Xi^2_T - \Xi^1_T \,\vert\, \cF_t] \\
      & + w^3_t \cdot \EE\left[ \int_t^T (\xi^2_u + \xi^1_u) \cdot
        K^1(t,u) \,du \,\Big\vert \,\cF_t \right] \\
      & + w^4_t \cdot \EE\left[ \int_t^T (\xi^2_u - \xi^1_u) \cdot
        K^2(t,u) \, du \, \Big\vert \, \cF_t \right]
    \end{aligned}
  \end{equation}
  with nonnegative kernels
  \begin{equation}
    \begin{aligned} \label{def:kernels}
      K^1(t,u) \set & \; \frac{w^1_t}{w^3_t}
      \frac{2\sigma e^{-\frac{\gamma}{6\lambda}(T-u)} \sinh
        (\sqrt{\delta^+}(T-u)/(3\lambda))}{\sqrt{\delta^+}}, \\
      K^2(t,u) \set & \; \frac{w^2_t}{w^4_t}
      \frac{2\sigma e^{\frac{\gamma}{2\lambda}(T-u)} \sinh
        (\sqrt{\delta^-}(T-u)/\lambda)}{\sqrt{\delta^-}}
    \end{aligned}
    \;\; (0 \leq t \leq u < T)
  \end{equation}
  which, for each $t \in [0,T)$, integrate to one over $[t , T]$.
  The solution $(\hat{X}^1, \hat{X}^2)$ of~\eqref{eq:Xsol} satisfies
  the terminal state constraints in the sense that
  \begin{equation} \label{eq:terminalconds}
    \lim_{t \uparrow T} \hat{X}^1_t = \Xi^1_T \quad \text{and} \quad
    \lim_{t \uparrow T} \hat{X}^2_t = \Xi^2_T  \quad \mathbb{P}\text{-a.s.}
  \end{equation}  
\end{Theorem}

The proof of Theorem~\ref{thm:main} consists of a verification that the pair $(\hat{\alpha}^1,\hat{\alpha}^2)$ with dynamics in~\eqref{eq:Xsol} is admissible (i.e., belongs to $\mathcal{A}^1\times\mathcal{A}^2$) and satisfies the FBSDE system in~\eqref{eq:FBSDE}. An explanation on how the Nash equilibrium $(\hat{\alpha}^1,\hat{\alpha}^2)$ can be constructed is provided in the appendix. 

\noindent\textbf{Proof of Theorem~\ref{thm:main}.} In view of
Lemma~\ref{lem:FOC} we merely have to show that
$(\hat{X}^1, \hat{X}^2, \hat{\alpha}^1, \hat{\alpha}^2)$ with dynamics
described in Theorem~\ref{thm:main}, equation~\eqref{eq:Xsol}, is a
solution of the FBSDE system in~\eqref{eq:FBSDE} with some suitable
square integrable martingales $(M^1_t)_{0 \leq t < T}$ and
$(M^2_t)_{0 \leq t < T}$. Uniqueness of the Nash equilibrium then
follows together with Lemma~\ref{lem:uniqueNE}.

\emph{Step 1:} We start with computing the dynamics of the controls
$\hat{\alpha}^1$ and $\hat{\alpha}^2$ in~\eqref{eq:Xsol} and verify that they satisfy the dynamics of the FBSDE system in~\eqref{eq:FBSDE}. To this end,
it is convenient to rewrite $w^1, w^2$ in~\eqref{def:weights}, as well
as $\hat{\xi}^1$ in~\eqref{def:optsignal1} and $\hat{\xi}^2$
in~\eqref{def:optsignal2} by introducing
\begin{align}
  \tilde{w}_t^1 \set
  & \; (c^+_t + c^-_t) w^1_t, \quad 
    \tilde{w}_t^2 \set (c^+_t + c^-_t) w^2_t  \quad (0 \leq t <
    T) \label{def:wtilde} \\
  \intertext{and}
  \tilde{\xi}_t^1 \set
  &\; (c^+_t + c^-_t) \hat{\xi}^1_t, \quad
    \tilde{\xi}_t^2 \set
    (c^+_t + c^-_t) \hat{\xi}^2_t  \quad (0 \leq t < T). \label{def:xitilde}
\end{align}
Moreover, setting
\begin{equation} \label{def:plusprocesses}
  \begin{aligned} 
    Y_t^+ \set & \; \int_0^t (\xi^1_s + \xi^2_s) \frac{2\sigma
    }{\sqrt{\delta^+}} e^{-\frac{\gamma}{6\lambda}(T-s)}
    \sinh(\sqrt{\delta^+}(T-s)/(3\lambda)) ds, \\
    M_t^+ \set & \; \EE\left[ \Xi^1_T + \Xi^2_T + Y_T^+ \,\vert\,
      \cF_t\right]
  \end{aligned}
\end{equation}
and
\begin{equation} \label{def:minusprocesses}
  \begin{aligned} 
    Y_t^- \set & \; \int_0^t (\xi^1_s - \xi^2_s) \frac{2\sigma
    }{\sqrt{\delta^-}} e^{\frac{\gamma}{2\lambda}(T-s)}
    \sinh(\sqrt{\delta^-}(T-s)/\lambda) ds, \hspace{27pt} \\
    M_t^- \set & \; \EE\left[ \Xi^1_T - \Xi^2_T + Y_T^- \,\vert\,
      \cF_t\right]
  \end{aligned}
\end{equation}
for all $0 \leq t \leq T$, we obtain the representations
\begin{equation} \label{eq:xitilde}
  \begin{aligned}
    \tilde{\xi}^1_t = &\; \tilde{w}^1_t ( M^+_t - Y^+_t ) +
    \tilde{w}^2_t ( M^-_t - Y^-_t ), \\
    \tilde{\xi}^2_t = &\; \tilde{w}^1_t ( M^+_t - Y^+_t ) -
    \tilde{w}^2_t ( M^-_t - Y^-_t )
  \end{aligned} \qquad (0 \leq t < T).
\end{equation}
In particular, 
\begin{equation} \label{eq:xitildediffs} \tilde{\xi}^1_t +
  \tilde{\xi}^2_t = 2 \tilde{w}^1_t ( M^+_t - Y^+_t ), \quad
  \tilde{\xi}^1_t - \tilde{\xi}^2_t = 2 \tilde{w}^2_t ( M^-_t -
    Y^-_t )
\end{equation}
on $[0,T)$. Note that $\Xi^1_T, \Xi^2_T, Y_T^+, Y_T^- \in L^2(\PP)$
implies that $(M_t^+)_{0 \leq t \leq T}$ and
$(M_t^-)_{0 \leq t \leq T}$ are square integrable martingales. Also,
observe that the processes
$Y^+, M^+, Y^-, M^- \in L^2(\PP \otimes dt)$. We can now
rewrite~\eqref{eq:Xsol} as
\begin{equation} \label{eq:XsolAlt}
  \begin{aligned} 
    \hat{\alpha}^1_t = &\; \frac{1}{2\lambda} ( \tilde{\xi}^1_t -
      c^+_t \hat{X}^2_t + c^-_t \hat{X}^2_t - c^+_t \hat{X}^1_t -
      c^-_t
      \hat{X}^1_t), \\
    \hat{\alpha}^2_t = &\; \frac{1}{2\lambda} ( \tilde{\xi}^2_t -
      c^+_t \hat{X}^1_t + c^-_t \hat{X}^1_t - c^+_t \hat{X}^2_t -
      c^-_t \hat{X}^2_t)
  \end{aligned} \qquad (0 \leq t < T).
\end{equation}
Next, for $\tilde{w}^1$, $\tilde{w}^2$ in~\eqref{def:wtilde} one can
easily check that
\begin{equation}
  (\tilde{w}^1_t)' = \tilde{w}^1_t \left(
    \frac{1}{\lambda} c^+_t - \frac{\gamma}{3\lambda} \right), \quad 
  (\tilde{w}^2_t)' = \tilde{w}^2_t \left(
    \frac{1}{\lambda} c^-_t + \frac{\gamma}{\lambda} \right)
  \quad (0 \leq t < T).
\end{equation}
Hence, by applying integration by parts in~\eqref{eq:xitilde} we
obtain the dynamics
\begin{equation} \label{eq:dynxitilde1}
  \begin{aligned}
    d\tilde{\xi}^1_t = &\; \tilde{w}^1_t ( M^+_t - Y^+_t )
    \left( \frac{1}{\lambda} c^+_t - \frac{\gamma}{3\lambda} \right)
    dt
    -\frac{2}{3} \sigma (\xi^1_t + \xi^2_t) dt  \\
    & + \tilde{w}^2_t ( M^-_t - Y^-_t ) \left(
      \frac{1}{\lambda} c^-_t + \frac{\gamma}{\lambda} \right) dt -
    2\sigma (\xi^1_t - \xi^2_t) dt \\
    & + \tilde{w}^1_t dM^+_t + \tilde{w}^2_t dM^-_t \qquad (0 \leq t <
    T)
  \end{aligned}
\end{equation}
and
\begin{equation} \label{eq:dynxitilde2}
  \begin{aligned}
    d\tilde{\xi}^2_t = &\; \tilde{w}^1_t ( M^+_t - Y^+_t )
    \left( \frac{1}{\lambda} c^+_t - \frac{\gamma}{3\lambda} \right)
    dt
    -\frac{2}{3} \sigma (\xi^1_t + \xi^2_t) dt  \\
    & - \tilde{w}^2_t ( M^-_t - Y^-_t ) \left(
      \frac{1}{\lambda} c^-_t + \frac{\gamma}{\lambda} \right) dt
    - 2
    \sigma (\xi^1_t - \xi^2_t) dt \\
    & + \tilde{w}^1_t dM^+_t - \tilde{w}^2_t dM^-_t \qquad (0 \leq t <
    T).
  \end{aligned}
\end{equation}
Now, having at hand~\eqref{eq:dynxitilde1} and~\eqref{eq:dynxitilde2},
as well as the fact that the functions $c^+, c^-$ in~\eqref{def:cpcm}
satisfy the ordinary Riccati differential equations
\begin{equation} \label{eq:riccati} (c^+_t)' =
  \frac{(c^+_t)^2}{\lambda} - \frac{\gamma}{3\lambda} c^+_t -
  \frac{2}{3} \sigma, \quad (c^-_t)' = \frac{(c^-_t)^2}{\lambda} +
  \frac{\gamma}{\lambda} c^-_t - 2\sigma \quad (0 \leq t < T),
\end{equation}
an elementary but tedious computation
reveals that the dynamics of $\hat{\alpha}^1$ and $\hat{\alpha}^2$
in~\eqref{eq:XsolAlt} on $[0,T)$ are given by
\begin{equation} \label{eq:dynalpha1}
  \begin{aligned}
    d\hat{\alpha}^1_t = & \; \hat{X}^1_t \left(
      \frac{4\sigma}{3\lambda} + \frac{\gamma}{6\lambda^2} c^+_t -
      \frac{\gamma}{2\lambda^2} c^-_t \right) dt -
    \frac{4\sigma}{3\lambda} \xi^1_t dt + \frac{\gamma}{6\lambda^2}
    \tilde{\xi}^1_t dt \\
    & + \hat{X}^2_t \left( -\frac{2\sigma}{3\lambda} +
      \frac{\gamma}{6\lambda^2} c^+_t + \frac{\gamma}{2\lambda^2} c^-_t
    \right) dt + \frac{2\sigma}{3\lambda} \xi^2_t dt -
    \frac{\gamma}{3\lambda^2}
    \tilde{\xi}^2_t dt \\
    & + \frac{\tilde{w}^1_t}{2\lambda} dM^+_t +
    \frac{\tilde{w}^2_t}{2\lambda} dM^-_t
  \end{aligned}
\end{equation}
and, similarly, by
\begin{equation} \label{eq:dynalpha2}
  \begin{aligned}
    d\hat{\alpha}^2_t = & \; \hat{X}^2_t \left(
      \frac{4\sigma}{3\lambda} + \frac{\gamma}{6\lambda^2} c^+_t -
      \frac{\gamma}{2\lambda^2} c^-_t \right) dt -
    \frac{4\sigma}{3\lambda} \xi^2_t dt + \frac{\gamma}{6\lambda^2}
    \tilde{\xi}^2_t dt \\
    & + \hat{X}^1_t \left( -\frac{2\sigma}{3\lambda} +
      \frac{\gamma}{6\lambda^2} c^+_t + \frac{\gamma}{2\lambda^2} c^-_t
    \right) dt + \frac{2\sigma}{3\lambda} \xi^1_t dt -
    \frac{\gamma}{3\lambda^2}
    \tilde{\xi}^1_t dt \\
    & + \frac{\tilde{w}^1_t}{2\lambda} dM^+_t -
    \frac{\tilde{w}^2_t}{2\lambda} dM^-_t,
  \end{aligned}
\end{equation}
where we also employed the identities in~\eqref{eq:xitildediffs}. As a
consequence, using the representations in~\eqref{eq:XsolAlt} we obtain
\begin{align*}
  & d\hat{\alpha}^1_t + \frac{1}{2} d\hat{\alpha}^2_t \\
  & =
    \frac{\sigma}{\lambda} (\hat{X}^1_t - \xi^1_t) dt -
    \frac{\gamma}{4\lambda^2} (\tilde{\xi}^2_t -
    c^+_t \hat{X}^1_t + c^-_t \hat{X}^1_t - c^+_t \hat{X}^2_t -
    c^-_t \hat{X}^2_t ) dt \\
  & \phantom{=}  + \frac{3}{4\lambda} \tilde{w}^1_t dM^+_t +
    \frac{1}{4\lambda} 
    \tilde{w}^2_t dM^-_t \\
  & = \frac{\sigma}{\lambda} (\hat{X}^1_t - \xi^1_t) dt
    - \frac{\gamma}{2\lambda} \hat{\alpha}^2_t dt + \frac{3}{4\lambda}
    \tilde{w}^1_t dM^+_t + \frac{1}{4\lambda}
    \tilde{w}^2_t dM^-_t \qquad (0 \leq t < T)
\end{align*}
and
\begin{align*}
  & d\hat{\alpha}^2_t + \frac{1}{2} d\hat{\alpha}^1_t \\
  & =
    \frac{\sigma}{\lambda} (\hat{X}^2_t - \xi^2_t) dt -
    \frac{\gamma}{4\lambda^2} (\tilde{\xi}^1_t -
    c^+_t \hat{X}^2_t + c^-_t \hat{X}^2_t - c^+_t \hat{X}^1_t -
    c^-_t \hat{X}^1_t ) dt \\
  & \phantom{=} + \frac{3}{4\lambda} \tilde{w}^1_t dM^+_t -
    \frac{1}{4\lambda} 
    \tilde{w}^2_t dM^-_t \\
  & = \frac{\sigma}{\lambda} (\hat{X}^2_t - \xi^2_t) dt
    - \frac{\gamma}{2\lambda} \hat{\alpha}^1_t dt + \frac{3}{4\lambda}
    \tilde{w}^1_t dM^+_t - \frac{1}{4\lambda}
    \tilde{w}^2_t dM^-_t \qquad (0 \leq t < T).
\end{align*}
In other words, the pair $(\hat{\alpha}^1, \hat{\alpha}^2)$ described
in~\eqref{eq:Xsol} satisfies the dynamics of the FBSDE system
in~\eqref{eq:FBSDE}, where $\int_0^\cdot \tilde{w}_t^{1} dM_t^{+}$,
$\int_0^\cdot \tilde{w}_t^{2} dM_t^{-}$ are square integrable
martingales on $[0,T)$ providing the ingredients for $M^1$ and $M^2$.

\emph{Step 2:} Next, we have to check the terminal conditions of the
FBSDE system in~\eqref{eq:FBSDE}, that is,
$\lim_{t \uparrow T} \hat{X}^1_t = \Xi^1_T$ and
$\lim_{t \uparrow T} \hat{X}^2_t = \Xi^2_T$ $\mathbb{P}$-a.s. holds
true for the pair of solutions $(\hat{X}^1, \hat{X}^2)$ of the coupled
ODE in~\eqref{eq:Xsol}. We adapt the argumentation
from~\citet{BankSonerVoss:17} which employs a simple comparison
principle for ordinary differential equations to our current
setting. Specifically, note that it suffices to show that
\begin{align}
  \lim_{t \uparrow T} (\hat{X}^1_t + \hat{X}^2_t) = & \; \Xi^1_T +
  \Xi^2_T\quad \PP\text{-a.s. and} \label{eq:plusLimit} \\
  \lim_{t \uparrow T} (\hat{X}^1_t - \hat{X}^2_t) = & \; \Xi^1_T - \Xi^2_T
  \quad \PP\text{-a.s.}, \label{eq:minusLimit}
\end{align}
where, using the dynamics in~\eqref{eq:Xsol} and the definition of
$w^5$ in~\eqref{def:weights}, the processes $\hat{X}^1 + \hat{X}^2$
and $\hat{X}^1 - \hat{X}^2$ satisfy, respectively, the ODE
\begin{equation}
  \begin{aligned} \label{eq:XodePlus}
    d(\hat{X}^1_t + \hat{X}^2_t) = &\; \frac{c^+_t + c^-_t}{2\lambda}
    \left( \hat{\xi}^1_t + \hat{\xi}^2_t - w^5_t \hat{X}^1_t - w^5_t
      \hat{X}^2_t -
      \hat{X}^1_t - \hat{X}^2_t \right) dt \\
    = & \; \frac{c^+_t}{\lambda} \left( \frac{\hat{\xi}^1_t +
        \hat{\xi}^2_t}{1+w^5_t} - (\hat{X}^1_t + \hat{X}^2_t) \right) dt
    \quad (0 \leq t <T)
  \end{aligned}
\end{equation}
and
\begin{equation} \label{eq:XodeMinus}
  \begin{aligned}
    d(\hat{X}^1_t - \hat{X}^2_t) = &\; \frac{c^+_t + c^-_t}{2\lambda}
    \left( \hat{\xi}^1_t - \hat{\xi}^2_t + w^5_t \hat{X}^1_t - w^5_t
      \hat{X}^2_t -
      \hat{X}^1_t + \hat{X}^2_t \right) dt \\
    = & \; \frac{c^-_t}{\lambda} \left( \frac{\hat{\xi}^1_t -
        \hat{\xi}^2_t}{1-w^5_t} - (\hat{X}^1_t - \hat{X}^2_t) \right) dt
    \quad (0 \leq t <T). 
  \end{aligned}
\end{equation}
Note that $w^5_t \in (-1,1)$ for all $t \in [0,T]$ by virtue of
Lemma~\ref{lem:weights} 1.). First, analogously to~\eqref{eq:xitilde}
let us rewrite $\hat{\xi}^1$ and $\hat{\xi}^2$
in~\eqref{def:optsignal1} and~\eqref{def:optsignal2} as
\begin{equation} \label{eq:xi}
  \begin{aligned}
    \hat{\xi}^1_t = &\; w^1_t ( M^+_t - Y^+_t ) +
    w^2_t ( M^-_t - Y^-_t ), \\
    \hat{\xi}^2_t = &\; w^1_t ( M^+_t - Y^+_t ) -
    w^2_t ( M^-_t - Y^-_t )
  \end{aligned} \qquad (0 \leq t \leq T)
\end{equation}
with $Y^+, M^+,Y^-,M^-$ as defined in~\eqref{def:plusprocesses}
and~\eqref{def:minusprocesses}. Hence, we can consider a
c\`adl\`ag version of the processes
$(\hat{\xi}^1_t)_{0 \leq t \leq T}$ and
$(\hat{\xi}^2_t)_{0 \leq t \leq T}$ and obtain, together with
Lemma~\ref{lem:weights}, 2.), the $\PP$-a.s. limits
\begin{equation*}
  \begin{aligned}
    \lim_{t \uparrow T} \hat{\xi}^1_t = & \; \frac{1}{2} \EE[\Xi^1_T +
    \Xi^2_T \, \vert \, \cF_{T-}] + \frac{1}{2} \EE[\Xi^1_T - \Xi^2_T
    \, \vert \, \cF_{T-}] = \Xi^1_T \quad \text{and} \\
    \lim_{t \uparrow T} \hat{\xi}^2_t = & \; \frac{1}{2} \EE[\Xi^1_T +
    \Xi^2_T \, \vert \, \cF_{T-}] - \frac{1}{2} \EE[\Xi^1_T - \Xi^2_T
    \, \vert \, \cF_{T-}] = \Xi^2_T
  \end{aligned}
\end{equation*}
due to $\cF_{T-}$-measurability of $\Xi^1_T$ and $\Xi^2_T$ by virtue
of our assumption in~\eqref{eq:assumption}. In particular, since
$\lim_{t\uparrow T} w^5_t = 0$ because of Lemma~\ref{lem:weights},
2.), it also holds that
\begin{equation} \label{eq:xilimits}
  \lim_{t \uparrow T} \frac{\hat{\xi}^1_t + \hat{\xi}^2_t}{1+w^5_t} = \Xi^1_T +
  \Xi^2_T \quad \ttext{and} \quad   \lim_{t \uparrow T} \frac{\hat{\xi}^1_t
  - \hat{\xi}^2_t}{1-w^5_t} = \Xi^1_T - 
  \Xi^2_T \quad \PP\text{-a.s.}
\end{equation}

Let us now start with proving the limit in~\eqref{eq:plusLimit}. As a
consequence of~\eqref{eq:xilimits}, for every $\varepsilon > 0$ there
exists a (random) time $\tau_\varepsilon \in [0,T)$ such that
$\PP$-a.s.
\begin{equation} \label{eq:epsilonBounds1}
\Xi^1_T + \Xi^2_T - \varepsilon \leq \frac{\hat{\xi}^1_t +
  \hat{\xi}^2_t}{1+w^5_t}  \leq \Xi^1_T + \Xi^2_T + \varepsilon \quad
\text{for all } t \in [\tau_\varepsilon ,T).
\end{equation}
Next, define $Y^{+,\varepsilon}_t \set \Xi^1_T + \Xi^2_T + \varepsilon -
(\hat{X}^1_t + \hat{X}^2_t)$ for all $t \in [0,T)$ so that
\begin{equation}
Y^{+,\varepsilon}_t \geq \frac{\hat{\xi}^1_t +
  \hat{\xi}^2_t}{1+w^5_t}  - (\hat{X}^1_t + \hat{X}^2_t) \quad
\text{for all } t \in [\tau_\varepsilon ,T).
\end{equation}
Together with the dynamics of $\hat{X}^1+\hat{X}^2$ in~\eqref{eq:XodePlus} this yields
\begin{equation} \label{eq:compODE1}
  \begin{aligned}
    d Y^{+,\varepsilon}_t = & \; -d(\hat{X}^1_t + \hat{X}^2_t) = -
    \frac{c^+_t}{\lambda} \left( \frac{\hat{\xi}^1_t +
        \hat{\xi}^2_t}{1+w^5_t} - (\hat{X}^1_t +
      \hat{X}^2_t) \right)  dt \\
    \geq & \; -\frac{c_t^+}{\lambda} Y^{+,\varepsilon}_t dt \quad
    \text{on } [\tau_\varepsilon ,T).
  \end{aligned}
\end{equation}
Moreover, since for all $\omega \in \Omega$ the linear ODE on
$[\tau_\varepsilon(\omega),T)$ given by
\begin{equation*}
Z^{+,\varepsilon}_{\tau_\varepsilon(\omega)} =
Y^{+,\varepsilon}_{\tau_\varepsilon(\omega)}(\omega), \quad dZ^{+,\varepsilon}_t =
-\frac{c_t^+}{\lambda} Z^{+,\varepsilon}_t dt 
\end{equation*}
admits the solution
\begin{align*}
  Z^{+,\varepsilon}_t  =  &\;
                            Y^{+,\varepsilon}_{\tau_\varepsilon(\omega)}(\omega)
                            \cdot 
                            e^{-\int_{\tau_\varepsilon}^{t} \frac{c^+_s}{\lambda} ds} \\
  = &\;
      Y^{+,\varepsilon}_{\tau_\varepsilon}(\omega) \cdot 
      e^{-\frac{\gamma}{6\lambda} (t-\tau_\varepsilon)} \cdot
      \frac{\sinh(\sqrt{\delta^+}(T-t)/(3\lambda))}
      {\sinh(\sqrt{\delta^+}(T-\tau_{\varepsilon})/(3\lambda))} \quad
      (\tau_\varepsilon \leq t < T)
\end{align*}
with $\lim_{t \uparrow T} Z^{+,\varepsilon}_t = 0$, the comparison
principle for ODEs in~\eqref{eq:compODE1} implies that
$Y^{+,\varepsilon}_t  \geq Z^{+,\varepsilon}_t$ for all $t \in
[\tau_\varepsilon, T)$ and thus
\begin{equation*}
\liminf_{t \uparrow T} Y^{+,\varepsilon}_t \geq \lim_{t \uparrow T}
Z^{+,\varepsilon}_t = 0 \quad \PP\text{-a.s.},
\end{equation*}
or, equivalently,
\begin{equation} \label{eq:limsup}
\limsup_{t \uparrow T} (\hat{X}^1_t + \hat{X}^2_t)  \leq \Xi^1_T +
\Xi^2_T + \varepsilon \quad \PP\text{-a.s.}
\end{equation}
Next, in a similar way, set
$\tilde{Y}^{+,\varepsilon}_t \set \Xi^1_T + \Xi^2_T - \varepsilon -
(\hat{X}^1_t + \hat{X}^2_t)$ for all $t \in [0,T)$ and observe as
above from~\eqref{eq:epsilonBounds1} that $\PP$-a.s. on
$[\tau_\varepsilon, T)$ it holds that
$d\tilde{Y}^{+,\varepsilon}_t \leq -\frac{c_t^+}{\lambda}
\tilde{Y}^{+,\varepsilon}_t dt$ and hence
\begin{equation*}
\limsup_{t \uparrow T} \tilde{Y}^{+,\varepsilon}_t \leq \lim_{t \uparrow T}
Z^{+,\varepsilon}_t \leq 0 \quad \PP\text{-a.s.}
\end{equation*}
by the comparison principle. That is,
\begin{equation*}
\liminf_{t \uparrow T} (\hat{X}^1_t + \hat{X}^2_t)  \geq \Xi^1_T +
\Xi^2_T - \varepsilon \quad \PP\text{-a.s.},
\end{equation*}
which, together with~\eqref{eq:limsup} yields the limit
in~\eqref{eq:plusLimit}. 

In fact, it can now be argued along the same lines as above that also
the limit in~\eqref{eq:minusLimit} holds true. Indeed, simply note
that~\eqref{eq:xilimits} implies similar to~\eqref{eq:epsilonBounds1}
that $\PP$-a.s. for every $\varepsilon > 0$ there exists a (random)
time $\tau'_\varepsilon \in [0,T)$ such that
\begin{equation*}
  \Xi^1_T - \Xi^2_T - \varepsilon \leq \frac{\hat{\xi}^1_t -
    \hat{\xi}^2_t}{1-w^5_t}  \leq \Xi^1_T - \Xi^2_T + \varepsilon \quad
  \text{for all } t \in [\tau'_\varepsilon ,T).
\end{equation*}
Then, introduce the processes
$Y^{-,\varepsilon}_t \set \Xi^1_T - \Xi^2_T + \varepsilon -
(\hat{X}^1_t - \hat{X}^2_t)$ and
$\tilde{Y}^{-,\varepsilon}_t \set \Xi^1_T - \Xi^2_T - \varepsilon -
(\hat{X}^1_t - \hat{X}^2_t)$ for all $t \in [0,T)$. By using the
dynamics of $\hat{X}^1 - \hat{X}^2$ in~\eqref{eq:XodeMinus} we can
once more apply the comparison principle on the interval
$[\tau'_\varepsilon,T)$ for the ODEs of $Y^{-,\varepsilon}$ and
$\tilde{Y}^{-,\varepsilon}$ together with the linear ODE
\begin{equation*}
  Z^{-,\varepsilon}_{\tau_\varepsilon} =
  z \in \mathbb{R}, \quad dZ^{-,\varepsilon}_t =
  -\frac{c_t^-}{\lambda} Z^{-,\varepsilon}_t dt, 
\end{equation*}
which admits the solution
\begin{equation*}
  Z^{-,\varepsilon}_t  = z
  e^{-\int_{\tau_\varepsilon'}^{t}
    \frac{c^-_s}{\lambda} ds} = z^-
  e^{\frac{\gamma}{2\lambda} (t-\tau_\varepsilon)}
  \frac{\sinh(\sqrt{\delta^-}(T-t)/\lambda)}
  {\sinh(\sqrt{\delta^-}(T-\tau_{\varepsilon}')/\lambda)} \quad
  (\tau'_\varepsilon \leq t < T)
\end{equation*}
such that $\lim_{t \uparrow T} Z^{-,\varepsilon}_t = 0$ to finally conclude
that
\begin{equation*}
\Xi^1_T -
\Xi^2_T - \varepsilon \leq \liminf_{t \uparrow T} (\hat{X}^1_t - \hat{X}^2_t) \leq \limsup_{t \uparrow T} (\hat{X}^1_t - \hat{X}^2_t) \leq \Xi^1_T -
\Xi^2_T + \varepsilon
\end{equation*}
as desired.

\emph{Step 3:} It is left to argue that the controls
$\hat{\alpha}^1, \hat{\alpha}^2$ described in~\eqref{eq:Xsol} belong
to the set $\cA$ in~\eqref{def:setA}, i.e.,
$\hat{\alpha}^1, \hat{\alpha}^2 \in L^2(\PP\otimes dt)$. To achieve
this we will follow a similar strategy as
in~\citet{BankSonerVoss:17}. For simplicity, we will assume without
loss of generality that $x^1=x^2=0$. Because of the coupling of
$\hat{\alpha}^1, \hat{\alpha}^2$ in~\eqref{eq:Xsol} it is more
convenient to prove that
$\hat{\alpha}^+ \set \hat{\alpha}^1 + \hat{\alpha}^2 \in
L^2(\PP\otimes dt)$ and
$\hat{\alpha}^- \set \hat{\alpha}^1 - \hat{\alpha}^2 \in
L^2(\PP\otimes dt)$, where we set
$\hat{X}^+_\cdot \set \int_0^\cdot \hat{\alpha}^+_s ds$ and
$\hat{X}^-_\cdot \set \int_0^\cdot \hat{\alpha}^-_s ds$. Recall
from~\eqref{eq:XodePlus} and~\eqref{eq:XodeMinus} above that we then
have
\begin{equation} 
  \hat{\alpha}^+_t = \frac{c^+_t}{\lambda} \left(
    \frac{\hat{\xi}^1_t + 
      \hat{\xi}^2_t}{1+w^5_t} - \hat{X}^+_t\right), \quad
  \hat{\alpha}^-_t = \frac{c^-_t}{\lambda} \left( \frac{\hat{\xi}^1_t -
      \hat{\xi}^2_t}{1-w^5_t} - \hat{X}^-_t \right) \label{eq:alphapm}
\end{equation}
on $[0,T)$, where
\begin{equation} \label{eq:xipm}
    \hat{\xi}^1_t + \hat{\xi}^2_t = 2 w^1_t (M^+_t - Y^+_t), \quad
    \hat{\xi}^1_t - \hat{\xi}^2_t = 2 w^2_t (M^-_t - Y^-_t) \quad (0 \leq t \leq T)
\end{equation}
because of~\eqref{eq:xi} (recall that $M^+,Y^+$ are given in~\eqref{def:plusprocesses} and $M^-,Y^-$ are given in~\eqref{def:minusprocesses}).

We start with showing that $\hat{\alpha}^+ \in L^2(\PP\otimes
dt)$. For this purpose, observe that it suffices to examine the following
two cases $\xi^1\equiv\xi^2\equiv 0$ and $\Xi^1_T=\Xi^2_T=0$ separately. Indeed, let us denote $\hat{\alpha}^{+,\xi^1,\xi^2,\Xi^1,\Xi^2} \triangleq \hat{\alpha}^{+}$ to emphasize also the dependence on $\xi^1,\xi^2,\Xi^1,\Xi^2$. Then, due to the linear dependence of $\hat{\alpha}^+$ in~\eqref{eq:alphapm} on $\xi^1,\xi^2,\Xi^1,\Xi^2$, it holds that
\begin{equation} \label{eq:alphaplusdec}
    \hat{\alpha}^{+,\xi^1,\xi^2,\Xi^1,\Xi^2} = \hat{\alpha}^{+,0,0,\Xi^1,\Xi^2} + \hat{\alpha}^{+,\xi^1,\xi^2,0,0}.
\end{equation}
Hence, it suffices to show that $\hat{\alpha}^{+,0,0,\Xi^1,\Xi^2} \in L^2(\PP\otimes
dt)$ and $\hat{\alpha}^{+,\xi^1,\xi^2,0,0} \in L^2(\PP\otimes
dt)$.

\underline{Case 1.1:} $\xi^1\equiv\xi^2\equiv 0$:

From~\eqref{eq:xipm} it follows that
$\hat{\xi}^1_t + \hat{\xi}^2_t = 2 w^1_t M^+_t$. Moreover, the
explicit solutions in~\eqref{eq:explSolX1} and~\eqref{eq:explSolX2}
yield
\begin{equation} \label{eq:Xplus}
  \begin{aligned}
    \hat{X}^+_t = & \; e^{-\int_0^t \frac{c^+_u}{\lambda} du} \int_0^t
    \frac{c_s^+ + c^-_s}{\lambda} w^1_s M^+_s e^{\int_0^s
      \frac{c^+_u}{\lambda} du} ds \\
    = & \; e^{\frac{\gamma}{6\lambda} (T-t)}
    \sinh(\sqrt{\delta^+}(T-t)/(3\lambda))  \\
    & \int_0^t M^+_s \frac{\sqrt{\delta^+}}{3\lambda
      \sinh(\sqrt{\delta^+}(T-s)/(3\lambda))^2} ds \qquad (0 \leq t < T).
  \end{aligned}
\end{equation}
Introducing the deterministic and differentiable function
$f^+_s \set 1/\sinh(\sqrt{\delta^+}(T-s)/(3\lambda))$ on $[0,T)$
allows to rewrite the integral in~\eqref{eq:Xplus} by applying
integration by parts as
\begin{align}
  & \int_0^t
    M^+_s \frac{\sqrt{\delta^+}}{3\lambda
    \sinh(\sqrt{\delta^+}(T-s)/(3\lambda))^2} ds = \int_0^t
    \tilde{M}^+_s df^+_s \nonumber \\
  & = \tilde{M}^+_t f^+_t - \tilde{M}^+_0 f^+_0 - \int_0^t f_s^+
    d\tilde{M}^+_s \qquad (0 \leq t < T), \label{eq:intMplus}
\end{align}
where $\tilde{M}^+_t \set
M_t^+/\cosh(\sqrt{\delta^+}(T-t)/(3\lambda))$ for all $t \in
[0,T)$. Moreover, we have that
\begin{equation} \label{eq:xiplus}
  \frac{\hat{\xi}^1_t + \hat{\xi}^2_t}{1+w^5_t} =
  \frac{\sqrt{\delta^+} e^{\frac{\gamma}{6\lambda}(T-t)}}{3 c^+_t
    \sinh(\sqrt{\delta^+}(T-t)/(3\lambda))} M^+_t \quad (0 \leq t \leq T).
\end{equation}
Now, plugging back~\eqref{eq:xiplus} and~\eqref{eq:Xplus} together
with~\eqref{eq:intMplus} into $\hat{\alpha}^+$ in~\eqref{eq:alphapm}
yields, after some elementary computations,
\begin{equation}
  \begin{aligned} \label{eq:alphaplus1}
    \hat{\alpha}^+_t =
    & \; -\frac{\gamma}{6\lambda}
    e^{\frac{\gamma}{6\lambda}(T-t)} \tilde{M}^+_t + \frac{c^+_t}{\lambda}
    e^{\frac{\gamma}{6\lambda} (T-t)}
    \sinh(\sqrt{\delta^+}(T-t)/(3\lambda)) \tilde{M}^+_0 f^+_0 \\
    & \; + \frac{c^+_t}{\lambda}
    e^{\frac{\gamma}{6\lambda} (T-t)}
    \sinh(\sqrt{\delta^+}(T-t)/(3\lambda)) \int_0^t f_s^+
    d\tilde{M}^+_s \quad (0 \leq t < T).
  \end{aligned}
\end{equation}
In fact, since $c^+_t \sinh(\sqrt{\delta^+}(T-t)/(3\lambda))$ is
bounded on $[0,T]$ (recall from~\eqref{def:cpcm} that $c^+_t = \frac{1}{3} \sqrt{\delta^+} \coth(\sqrt{\delta^+}(T-t)/(3\lambda)) + \frac{1}{6}\gamma$) and $\tilde{M}^+ \in L^2(\PP \otimes dt)$ (recall
that $M^+$ in~\eqref{def:plusprocesses} belongs to
$L^2(\PP \otimes dt)$) the first two terms in~\eqref{eq:alphaplus1}
are in $L^2(\PP \otimes dt)$. For the stochastic integral, we obtain
\begin{align*}
  \int_0^t f_s^+ d\tilde{M}^+_s =
  & \; \int_0^t \frac{\sqrt{\delta^+}
    M_s^+}{3\lambda
    \cosh(\sqrt{\delta^+}(T-s)/(3\lambda))^2}
    ds                                  
  \\  
  & \; + \int_0^t
    \frac{\tilde{f}_s^+}{\cosh(\sqrt{\delta^+}(T-s)/(3\lambda))} dM^+_s, 
\end{align*}
where the first integral on the right is again an element of
$L^2(\PP \otimes dt)$. The second integral satisfies
\begin{equation} \label{eq:integralMplus}
  \begin{aligned}
    & \EE \left[ \int_0^T \left( \int_0^t
        \frac{\tilde{f}_s^+}{\cosh(\sqrt{\delta^+}(T-s)/(3\lambda))}
        dM^+_s \right)^2 dt \right] \\
    & = \EE \left[ \int_0^T \int_0^t \left(
        \frac{\tilde{f}_s^+}{\cosh(\sqrt{\delta^+}(T-s)/(3\lambda))}
      \right)^2
      d\langle M^+ \rangle_s dt \right] \\
    & = \EE \left[ \int_0^T (T-s)
      \frac{(\tilde{f}_s^+)^2}{\cosh(\sqrt{\delta^+}(T-s)/(3\lambda))^2}
      d\langle M^+ \rangle_s \right] \\
    & \leq \frac{9\lambda^2}{\delta^+} \EE \left[ \int_0^T
      \frac{1}{T-s} d\langle M^+ \rangle_s \right] < \infty
  \end{aligned}
\end{equation}
by our assumption in~\eqref{eq:assumption}, where we also used
Fubini's theorem twice and the fact that $\sinh(\tau) \geq \tau$ and
$\cosh(\tau) \geq 1$ for all $\tau \geq 0$. That is, we obtain that
$\hat{\alpha}^+ \in L^2(\PP \otimes dt)$ in this case.

\underline{Case 1.2:} $\Xi^1_T=\Xi^2_T=0$:

In this case, we obtain from the expressions in~\eqref{def:optsignal1}
and~\eqref{def:optsignal2} that
\begin{equation*}
  \hat{\xi}^1_t + \hat{\xi}^2_t = 2 w^3_t \EE \left[ \int_t^T (\xi^1_u +
    \xi^2_u) K^1(t,u) du \, \Big\vert \, \cF_t \right] \quad (0 \leq t
  \leq T)
\end{equation*}
and thus, using again the explicit representation for
$\hat{X}^+=\hat{X}^1 + \hat{X}^2$ from~\eqref{eq:explSolX1}
and~\eqref{eq:explSolX2}, $\hat{\alpha}^+$ in~\eqref{eq:alphapm}
becomes
\begin{align}
  &\hat{\alpha}^+_t = \frac{c^+_t}{\lambda} \left(
    \frac{\hat{\xi}^1_t + 
    \hat{\xi}^2_t}{1+w^5_t} - \hat{X}^+_t\right) \nonumber \\
  & = \frac{2c^+_tw^3_t}{\lambda (1+w^5_t)} \EE \left[ \int_t^T (\xi^1_u +
    \xi^2_u) K^1(t,u) du \, \bigg\vert \, \cF_t \right] \nonumber \\
  & \quad - \frac{c^+_t}{\lambda} e^{-\int_0^t \frac{c^+_u}{\lambda} du}
  \nonumber \\
  & \qquad \int_0^t \frac{(c_s^+ + c^-_s)w^3_s}{\lambda} e^{\int_0^s
    \frac{c^+_u}{\lambda} du} \EE \left[ \int_s^T (\xi^1_u +
    \xi^2_u) K^1(s,u) du \, \bigg\vert \, \cF_s \right]  ds. \label{eq:alphaplus2}
\end{align}
In fact, it holds that all the ratios in~\eqref{eq:alphaplus2}
involving $c^+$, $c^-$ are bounded on $[0,T]$. Moreover, by
Lemma~\ref{lem:bounds} we have
\begin{equation*}
  \EE \left[ \int_t^T (\xi^1_u +
    \xi^2_u) K^1(t,u) du \, \bigg\vert \, \cF_t \right]  \in L^2(\PP
  \otimes dt),
\end{equation*}
as well as
\begin{align*}
  & \EE \left[ \int_0^T \left(  \int_0^t \EE \left[ \int_s^T (\xi^1_u +
  \xi^2_u) K^1(s,u) du \, \bigg\vert \, \cF_s \right] ds \right)^2 dt
  \right] \\
  & \leq \frac{T^2}{2} \EE \left[ \int_0^T \left(  \EE \left[ \int_s^T (\xi^1_u +
  \xi^2_u) K^1(s,u) du \, \bigg\vert \, \cF_s \right] \right)^2 ds
  \right] < \infty
\end{align*}
by using Jensen's inequality. As a consequence, we can also conclude
in this case that $\hat{\alpha}^+$ belongs to $L^2(\PP \otimes dt)$.

Let us now argue that also $\hat{\alpha}^-$ in~\eqref{eq:alphapm}
belongs to $L^2(\PP\otimes dt)$. The argumentation is very similar to
the one presented above so that we only sketch the main steps. Again,
it is enough to investigate the following
two cases $\xi^1\equiv\xi^2\equiv 0$ and $\Xi^1_T=\Xi^2_T=0$ separately because $\hat{\alpha}^-$ in~\eqref{eq:alphapm} can similarly be  decomposed as $\hat{\alpha}^+$ in~\eqref{eq:alphaplusdec}.

\underline{Case 2.1:} $\xi^1\equiv\xi^2\equiv 0$:

Similar to~\eqref{eq:Xplus} above, using
$\hat{\xi}^1_t - \hat{\xi}^2_t = 2 w^2_t M^-_t$ from~\eqref{eq:xipm} we
obtain via~\eqref{eq:explSolX1} and~\eqref{eq:explSolX2} the
representation
\begin{equation} \label{eq:Xminus}
  \begin{aligned}
    \hat{X}^-_t = & \; e^{-\int_0^t \frac{c^-_u}{\lambda} du} \int_0^t
    \frac{c_s^+ + c^-_s}{\lambda} w^2_s M^-_s e^{\int_0^s
      \frac{c^-_u}{\lambda} du} ds \\
    = & \; e^{-\frac{\gamma}{2\lambda} (T-t)}
    \sinh(\sqrt{\delta^-}(T-t)/\lambda)  \\
    & \int_0^t M^-_s \frac{\sqrt{\delta^-}}{\lambda
      \sinh(\sqrt{\delta^-}(T-s)/\lambda)^2} ds \qquad (0 \leq t < T).
  \end{aligned}
\end{equation}
Setting $f^-_s \set 1/\sinh(\sqrt{\delta^-}(T-s)/\lambda)$ on $[0,T)$
we can rewrite the integral in~\eqref{eq:Xminus} as
\begin{equation}
  \int_0^t
  \tilde{M}^-_s df^-_s = \tilde{M}^-_t f^-_t - \tilde{M}^-_0 f^-_0 -
  \int_0^t f_s^- 
  d\tilde{M}^-_s \qquad (0 \leq t < T) \label{eq:intMminus}
\end{equation}
with $\tilde{M}^-_t \set M_t^-/\cosh(\sqrt{\delta^-}(T-t)/\lambda)$
for all $t \in [0,T)$. In addition,
\begin{equation} \label{eq:ximinus} \frac{\hat{\xi}^1_t -
    \hat{\xi}^2_t}{1-w^5_t} = \frac{\sqrt{\delta^-}
    e^{-\frac{\gamma}{2\lambda}(T-t)}}{c^-_t
    \sinh(\sqrt{\delta^-}(T-t)/\lambda)} M^-_t \quad (0 \leq t \leq
  T).
\end{equation}
Inserting~\eqref{eq:ximinus} and~\eqref{eq:Xminus} together
with~\eqref{eq:intMminus} into $\hat{\alpha}^-$ in~\eqref{eq:alphapm}
then yields
\begin{equation}
  \begin{aligned} \label{eq:alphaminus1}
    \hat{\alpha}^-_t =
    & \; \frac{\gamma}{2\lambda}
    e^{-\frac{\gamma}{2\lambda}(T-t)} \tilde{M}^-_t + \frac{c^-_t}{\lambda}
    e^{-\frac{\gamma}{2\lambda} (T-t)}
    \sinh(\sqrt{\delta^-}(T-t)/\lambda) \tilde{M}^-_0 f^-_0 \\
    & \; + \frac{c^-_t}{\lambda}
    e^{-\frac{\gamma}{2\lambda} (T-t)}
    \sinh(\sqrt{\delta^-}(T-t)/\lambda) \int_0^t f_s^-
    d\tilde{M}^-_s \quad (0 \leq t < T),
  \end{aligned}
\end{equation}
where
\begin{align}
  \int_0^t f_s^- d\tilde{M}^-_s =
  & \; \int_0^t \frac{\sqrt{\delta^-}
    M_s^-}{\lambda
    \cosh(\sqrt{\delta^-}(T-s)/\lambda)^2}
    ds                                  
  \nonumber \\  
  & \; + \int_0^t
    \frac{\tilde{f}_s^-}{\cosh(\sqrt{\delta^-}(T-s)/\lambda)}
    dM^-_s. \label{eq:stochIntMminus} 
\end{align}
Observe as in~\eqref{eq:alphaplus1} above that
$c^-_t \sinh(\sqrt{\delta^-}(T-t)/\lambda)$ is bounded on $[0,T]$ (recall from~\eqref{def:cpcm} that $c^-_t = \sqrt{\delta^-}
    \coth(\sqrt{\delta^-}(T-t)/\lambda)-\frac{1}{2}\gamma$) and
that $\tilde{M}^- \in L^2(\PP \otimes dt)$. Therefore, we only need to
justify that the stochastic integral in~\eqref{eq:stochIntMminus}
belongs to $L^2(\PP \otimes dt)$. Indeed, by the same computations as
in~\eqref{eq:integralMplus}, we obtain via our assumption
in~\eqref{eq:assumption} that
\begin{equation} \label{eq:integralMminus}
  \begin{aligned}
    & \EE \left[ \int_0^T \left( \int_0^t
        \frac{\tilde{f}_s^-}{\cosh(\sqrt{\delta^-}(T-s)/\lambda)}
        dM^-_s \right)^2 dt \right] \\
    & \leq \frac{\lambda^2}{\delta^-} \EE \left[ \int_0^T
      \frac{1}{T-s} d\langle M^- \rangle_s \right] < \infty.
  \end{aligned}
\end{equation}
Hence, we can conclude that $\hat{\alpha}^- \in L^2(\PP \otimes dt)$
in this case. 

\underline{Case 2.2:} $\Xi^1_T=\Xi^2_T=0$:

Here, similar to~\eqref{eq:alphaplus2} above,~\eqref{def:optsignal1}
and~\eqref{def:optsignal2} imply that
\begin{equation*}
  \hat{\xi}^1_t - \hat{\xi}^2_t = 2 w^4_t \EE \left[ \int_t^T (\xi^1_u -
    \xi^2_u) K^2(t,u) du \, \Big\vert \, \cF_t \right] \quad (0 \leq t
  \leq T)
\end{equation*}
and hence, together with $\hat{X}^-=\hat{X}^1 - \hat{X}^2$
from~\eqref{eq:explSolX1} and~\eqref{eq:explSolX2}, $\hat{\alpha}^-$
in~\eqref{eq:alphapm} can be written as
\begin{align}
  &\hat{\alpha}^-_t = \frac{c^-_t}{\lambda} \left(
    \frac{\hat{\xi}^1_t - 
    \hat{\xi}^2_t}{1-w^5_t} - \hat{X}^-_t\right) \nonumber \\
  & = \frac{2c^-_tw^4_t}{\lambda (1-w^5_t)} \EE \left[ \int_t^T (\xi^1_u -
    \xi^2_u) K^2(t,u) du \, \bigg\vert \, \cF_t \right] \nonumber \\
  & \quad - \frac{c^-_t}{\lambda} e^{-\int_0^t \frac{c^-_u}{\lambda} du}
    \nonumber \\
  & \qquad \int_0^t \frac{(c_s^+ + c^-_s)w^4_s}{\lambda} e^{\int_0^s
    \frac{c^-_u}{\lambda} du} \EE \left[ \int_s^T (\xi^1_u -
    \xi^2_u) K^2(s,u) du \, \bigg\vert \, \cF_s \right]  ds. \label{eq:alphaminus2}
\end{align}
As in~\eqref{eq:alphaplus2}, all the ratios in~\eqref{eq:alphaminus2}
involving the functions $c^+$, $c^-$ are bounded on $[0,T]$, and we
can conclude along the same lines as in step 2.1 by virtue of
Lemma~\ref{lem:bounds} that $\hat{\alpha}^- \in L^2(\PP \otimes dt)$
in this case as well.

\emph{Step 4:} Finally, we have to argue that the functions $K^1(t,u)$ and $K^2(t,u)$ defined in~\eqref{def:kernels} are nonnegative kernels which integrate to one over $[t,T)$ as functions in $u \in [t,T)$. To this end, observe that $c^+_t > 0$ and $c^-_t > 0$ for
all $t \in [0,T]$, which implies that $w^1_\cdot, w^2_\cdot > 0$ on
$[0,T)$. Moreover, a direct computation yields that for all
$t \in [0,T)$ we have
\begin{equation} \label{eq:kernelIntegral} 
  \begin{aligned}
    0 < & \; \int_t^T \frac{2\sigma }{\sqrt{\delta^+}}
    e^{-\frac{\gamma}{6\lambda}(T-u)}
    \sinh(\sqrt{\delta^+}(T-u)/(3\lambda)) du =
    \frac{w^3_t}{w^1_t}, \\
    0 < & \; \int_t^T \frac{2\sigma}{\sqrt{\delta^-}}
    e^{\frac{\gamma}{2\lambda}(T-u)}
    \sinh(\sqrt{\delta^-}(T-u)/\lambda) du = \frac{w^4_t}{w^2_t}.
  \end{aligned}
\end{equation}
Thus, we also obtain that $w^3_\cdot, w^4_\cdot > 0$ on $[0,T)$. But
this implies for the functions defined in~\eqref{def:kernels} that
$K^1(t,u) > 0$ and $K^2(t,u) > 0$ for all $0 \leq t \leq u < T$, as
well as that $\int_t^T K^1(t,u) du = \int_t^T K^2(t,u) du = 1$ for all
$t \in [0,T)$. \qed \medskip

The equilibrium share holdings prescribed by
the linear coupled ODE in~\eqref{eq:Xsol} can also be computed explicitly.

\begin{Corollary} \label{cor:Xsol}
  The solution $(\hat{X}^1, \hat{X}^2)$ to the linear ODE
  in~\eqref{eq:Xsol} is given by
    \begin{align}
      \hat{X}^{1}_t = & \; \frac{1}{2} (x^1 + x^2) e^{-\int_0^t
        \frac{c^+_s}{\lambda} ds} + \frac{1}{4\lambda} \int_0^t (c^+_s
      + c^-_s) (\hat{\xi}^1_s+\hat{\xi}^2_s) e^{-\int_s^t
        \frac{c^+_u}{\lambda} du} ds \nonumber \\
      & \; + \frac{1}{2} (x^1 - x^2) e^{-\int_0^t
        \frac{c^-_s}{\lambda} ds} + \frac{1}{4\lambda} \int_0^t
      (c^+_s + c^-_s) (\hat{\xi}^1_s-\hat{\xi}^2_s) e^{-\int_s^t
        \frac{c^-_u}{\lambda} du} ds \label{eq:explSolX1}  \\
      \intertext{and, similarly, by}
      \hat{X}^{2}_t = & \; \frac{1}{2} (x^2 + x^1) e^{-\int_0^t
        \frac{c^+_s}{\lambda} ds} + \frac{1}{4\lambda} \int_0^t (c^+_s
      + c^-_s) (\hat{\xi}^2_s+\hat{\xi}^1_s) e^{-\int_s^t
        \frac{c^+_u}{\lambda} du} ds \nonumber \\
      & \; + \frac{1}{2} (x^2 - x^1) e^{-\int_0^t
        \frac{c^-_s}{\lambda} ds} + \frac{1}{4\lambda} \int_0^t
      (c^+_s + c^-_s) (\hat{\xi}^2_s-\hat{\xi}^1_s) e^{-\int_s^t
        \frac{c^-_u}{\lambda} du} ds  \label{eq:explSolX2}
    \end{align}
  for all $t \in [0,T]$.
\end{Corollary}

\begin{proof}
Recall that from
the dynamics of $\hat{X}^1$ and $\hat{X}^2$ 
in~\eqref{eq:Xsol} we obtain that the processes
$\hat{X}^1 + \hat{X}^2$ and $\hat{X}^1 - \hat{X}^2$ satisfy,
respectively, the linear ODEs in~\eqref{eq:XodePlus}
and~\eqref{eq:XodeMinus} with initial values $x^1+x^2$ and
$x^1-x^2$. Applying the variation of constants formula then yields
\begin{equation*}
  \hat{X}^1_t \pm \hat{X}^2_t = (x^1 \pm x^2) e^{-\int_0^t
    \frac{c^{\pm}_s}{\lambda} ds} + \int_0^t
  \frac{c^+_s+c^-_s}{2\lambda} (\hat{\xi}^1_s \pm \hat{\xi}^2_s)
  e^{-\int_s^t \frac{c_u^{\pm}}{\lambda} du} ds
\end{equation*}
and hence the assertion in~\eqref{eq:explSolX1}
and~\eqref{eq:explSolX2} via the obvious relation
\begin{equation*}
  \hat{X}^{1,2}_t = \frac{1}{2} (\hat{X}^1_t + \hat{X}^2_t) \pm
  \frac{1}{2} (\hat{X}^1_t - \hat{X}^2_t).
\end{equation*}
\end{proof} 

Lastly, following simple properties of the weight functions introduced in~\eqref{def:weights} will help enlightening the structure of the Nash equilibrium presented in Theorem~\ref{thm:main}.

\begin{Lemma} \label{lem:weights}
  The weight functions $w^1, w^2, w^3, w^4,w ^5$ defined
  in~\eqref{def:weights} satisfy
  \begin{enumerate}
  \item $w_\cdot^5 \in (-1,1)$, $w_{\cdot}^{1,2,3,4} > 0$ on $[0,T)$
    and $w^1_\cdot + w^2_\cdot + w^3_\cdot + w^4_\cdot =1$ on $[0,T]$,
  \item $\lim_{t \uparrow T} w_t^{1,2} = 1/2$ and
    $\lim_{t \uparrow T} w_t^{3,4,5} = 0$.
  \end{enumerate}
\end{Lemma}

\begin{proof} 
\emph{1.} First,
recall from the proof of Theorem~\ref{thm:main}, Step~4, above that
$w^1_\cdot, w^2_\cdot,w^3_\cdot,w^4_\cdot > 0$ on $[0,T)$. Moreover,
from the definition in~\eqref{def:weights} we immediately obtain that
$w^1_t + w^2_t + w^3_t + w^4_t =1$ for all $t \in[0,T]$. Together with
the fact that $c^+_\cdot > 0$ and $c^-_\cdot > 0$ on $[0,T]$, we
also observe that $w^5_t \in (-1,1)$ for all $t \in [0,T]$.

\emph{2.} Concerning the limiting behaviour of the weight functions,
it suffices to note that
\begin{equation*}
  \lim_{t \uparrow T}
  \frac{\sinh(\sqrt{\delta^+}(T-t)/(3\lambda))}{\sinh(\sqrt{\delta^-}(T-t)/\lambda)}
  = \frac{\sqrt{\delta^+}}{3\sqrt{\delta^-}}.  
\end{equation*}
Then, rewriting $w^1$, $w^2$ in~\eqref{def:weights} by plugging in
$c^+$, $c^-$ from~\eqref{def:cpcm} to obtain the representations
\begin{equation*}
  w^1_t = \frac{\sqrt{\delta^+}
    e^{\frac{\gamma}{6\lambda}(T-t)}}{d^1_t}, \qquad  w^2_t =
  \frac{3 \sqrt{\delta^-}
    e^{-\frac{\gamma}{2\lambda}(T-t)}}{d^2_t}
\end{equation*}
with 
\begin{align*}
  d^1_t \set & \;
               \sqrt{\delta^+}\cosh(\sqrt{\delta^+}(T-t)/(3\lambda))-\gamma 
               \sinh(\sqrt{\delta^+}(T-t)/(3\lambda)) \\
             & + \sqrt{\delta^-}
               \sinh(\sqrt{\delta^+}(T-t)/(3\lambda))
               \coth(\sqrt{\delta^-}(T-t)/\lambda),  \\
  d^2_t \set & \; 3 \sqrt{\delta^-}
               \cosh(\sqrt{\delta^-}(T-t)/\lambda) -\gamma
               \sinh(\sqrt{\delta^-}(T-t)/\lambda) \\
             & + \sqrt{\delta^+}\sinh(\sqrt{\delta^-}(T-t)/\lambda)
               \coth(\sqrt{\delta^+}(T-t)/(3\lambda)) 
\end{align*}
yields
\begin{equation*}
  \lim_{t\uparrow T} w^1_t =
  \frac{\sqrt{\delta^+}}{\sqrt{\delta^+} + \sqrt{\delta^+}} =
  \frac{1}{2}, \qquad \lim_{t\uparrow T} w^2_t
  = \frac{\sqrt{\delta^-}}{\sqrt{\delta^-}
    + \sqrt{\delta^-}} = \frac{1}{2}.  
\end{equation*}
Similarly, with
\begin{align*}
  \frac{c_t^+}{c^+_t+c^-_t} =
  & \; \frac{2\sqrt{\delta^+}
    \coth(\sqrt{\delta^+}(T-t)/(3\lambda)) +
    \gamma}{2\sqrt{\delta^+}
    \coth(\sqrt{\delta^+}(T-t)/(3\lambda)) +
    6\sqrt{\delta^-} \coth(\sqrt{\delta^-}(T-t)/\lambda)-2\gamma} \\
  \frac{c_t^-}{c^+_t+c^-_t} =
  & \; \frac{6 \sqrt{\delta^-}
    \coth(\sqrt{\delta^-}(T-t)/\lambda) - 3
    \gamma}{2\sqrt{\delta^+}
    \coth(\sqrt{\delta^+}(T-t)/(3\lambda)) +
    6\sqrt{\delta^-} \coth(\sqrt{\delta^-}(T-t)/\lambda)-2\gamma}
\end{align*}
we also have
\begin{equation*}
  \lim_{t\uparrow T} \frac{c_t^+}{c^+_t+c^-_t} =
  \frac{\sqrt{\delta^+}}{\sqrt{\delta^+} + \sqrt{\delta^+}} =
  \frac{1}{2}, \qquad \lim_{t\uparrow T} \frac{c_t^-}{c^+_t+c^-_t} =
  \frac{\sqrt{\delta^-}}{\sqrt{\delta^-} + \sqrt{\delta^-}} = \frac{1}{2}
\end{equation*}
and hence
\begin{equation*}
 \lim_{t\uparrow T} w^3_t =  \lim_{t\uparrow T} w^4_t =
 \lim_{t\uparrow T} w^5_t = 0
\end{equation*}
as desired.
\end{proof}

The final lemma provides estimates with respect to the
$L^2(\PP\otimes dt)$-norm which are used in the proof of
Theorem~\ref{thm:main} above.

\begin{Lemma} \label{lem:bounds} Let
  $(\zeta_t)_{0 \leq t \leq T} \in L^2(\PP\otimes dt)$ be
  progressively measurable. Moreover, let $K^1(t,u)$, $K^2(t,u)$,
  $0 \leq t \leq u < T$, denote the kernels from
  Theorem~\ref{thm:main}.
  \begin{itemize}
  \item[a)] For
    $\zeta^{K^1}_t \set \EE[ \int_t^T \zeta_u K^1(t,u) du \vert
    \cF_t]$, $0 \leq t < T$, it holds that
    \begin{equation*}
      \Vert \zeta^{K^1} \Vert_{L^2(\PP\otimes dt)} \leq c
      \Vert \zeta \Vert_{L^2(\PP\otimes dt)}  
    \end{equation*}
    for some constant $c>0$.
  \item[b)] For
    $\zeta^{K^2}_t \set \EE[ \int_t^T \zeta_u K^2(t,u) du \vert
    \cF_t]$, $0 \leq t < T$, it holds that
    \begin{equation*}
      \Vert \zeta^{K^2} \Vert_{L^2(\PP\otimes dt)} \leq c
      \Vert \zeta \Vert_{L^2(\PP\otimes dt)}  
    \end{equation*}
    for some constant $c>0$.
  \end{itemize}
\end{Lemma}

\begin{proof}
  Both upper bounds can be verified in a similar fashion as in the
  proof of Lemma 5.5 in~\citet{BankSonerVoss:17}. We will thus omit it
  here.
\end{proof}

\begin{Remark} \label{rem:solution} 
Following up on
  Remark~\ref{rem:problem}, setting $\xi^1 \equiv \xi^2 \equiv 0$
  and $\Xi^1_T = \Xi^2_T = 0$ $\PP$-almost surely, our
  Theorem~\ref{thm:main} together with Corollary~\ref{cor:Xsol}
  retrieves the two-player results from~\citet[Result
  1]{CarlinLoboViswanathan:07} for the case $\sigma = 0$ and
  from~\citet[Corollary 2.6]{SchiedZhang:17} for the case
  $\sigma > 0$. Note that this configuration yields
  $\hat{\xi}^1 \equiv \hat{\xi}^2 \equiv 0$ in~\eqref{def:optsignal1}
  and~\eqref{def:optsignal2}, which in turn implies that the Nash
  equilibrium trading rates in~\eqref{eq:Xsol} and the corresponding
  share holdings in~\eqref{eq:explSolX1} and~\eqref{eq:explSolX2} are
  deterministic.
\end{Remark}

We end this section by briefly discussing qualitatively the Nash equilibrium obtained in Theorem~\ref{thm:main}. Very similar to the single-player solution in~\cite{BankSonerVoss:17} it turns out that the trading rates $\hat{\alpha}^1$ and $\hat{\alpha}^2$ in~\eqref{eq:Xsol} prescribe, respectively, to gradually trade in the direction of an optimal signal process $\hat{\xi}^1_t$ and $\hat{\xi}^2_t$ (rather than toward the actual target position $\xi^1_t$, $\xi^2_t$), which is further adjusted by a fraction $w^5_t \in (-1,1)$ of the opponent's respective current portfolio position $\hat{X}^2_t$ and $\hat{X}^1_t$. The optimal signal processes $\hat{\xi}^1$ in~\eqref{def:optsignal1} and $\hat{\xi}^2$ in~\eqref{def:optsignal2} are convex combinations of weighted averages of expected future target positions of the processes $\xi^1$, $\xi^2$ and the expected terminal positions $\Xi^1_T$, $\Xi^2_T$, where the weights $w^1_t, w^2_t, w^3_t, w^4_t$ systematically shift toward the desired individual terminal state as $t \uparrow T$ (Lemma~\ref{lem:weights} implies that $\lim_{t \uparrow T} \hat{\xi}^i_t = \Xi^i_T$ $\PP$-a.s. for both players $i=1,2$). The increasing urgency rate $(c^+_t+c^-_t)/(2\lambda) \uparrow \infty$ for $t \uparrow T$, together with $\lim_{t\uparrow T} w^5_t = 0$, then forces both strategies in~\eqref{eq:Xsol} to end up in the predetermined terminal portfolio position at maturity~$T$ (see also the proof of Theorem~\ref{thm:main} above). Interestingly, we note that the first agent's optimal signal process $\hat{\xi}^1$ not only seeks to anticipate the future evolution of her own target strategy $\xi^1$ but, conscious of her competitor's trading goals, does so also for the opponent's target strategy~$\xi^2$. In other words, besides following her own objectives, she also takes into account the other agent's known trading intentions. Moreover, the weights $w^3_t$ and $w^4_t$ dictate the actual \emph{trading direction} with respect to the other agent's tracking target. Indeed, observe that if $w^3_t$ predominates $w^4_t$ in~\eqref{def:optsignal1}, the first player's optimal signal $\hat{\xi}^1$ directs to also trade in parallel in the \emph{same} direction as the second player, that is, in the direction of the expected future average positions of $\xi^2$. In contrast, if $w^4_t$ outweighs $w^3_t$, then the optimal signal imposes to trade in the \emph{opposite} direction of the second player's target strategy, i.e., toward the expected weighted averages of $-\xi^2$. The former case can be viewed as a \emph{predatory} trading action of the first agent against the second agent, whereas the latter case can be regarded as a \emph{cooperative} behaviour. The same applies for the second player
in~\eqref{def:optsignal2} due to symmetry. In our illustrations in Section~\ref{sec:illustrations} below it becomes apparent that both these cases depend on the relationship between the permanent and temporary price impact parameters $\gamma$ and $\lambda$. Loosely speaking, in a plastic market where $\gamma \gg \lambda$, the weight $w^3$ tends to be larger than $w^4$, and in an elastic market with $\lambda \gg \gamma$ we have that $w^4$ tends to be larger than $w^3$ (see also the graphical illustration of the weight functions in Figure~\ref{fig:weights} below). In this regard, depending on the illiquidity parameters the optimal signal processes~$\hat{\xi}^1$ and~$\hat{\xi}^2$ account for different types of regimes. It turns out that this leads to qualitative different behavioral patterns in the Nash equilibrium where both predation and cooperation between the agents can occur, even in a coexisting manner.

\section{Illustrations} \label{sec:illustrations}

In this section we present some case studies to illustrate the
qualitative behaviour of the two-player Nash equilibrium presented in
Theorem~\ref{thm:main}.

\subsection{Optimal liquidation revisited}

We start with revisiting the differential game of optimal portfolio
liquidation studied in~\citet{SchiedZhang:17}. Specifically, the first
agent seeks to liquidate her initial portfolio position of $x^1=1$
shares in the risky asset by time $T=2$ and hence requires her
terminal position to satisfy $\Xi^1_T=0$ $\PP$-a.s. at final
time. Vigilant about her share holdings and in line with her selling
intention she also wants her inventory to be close to 0 throughout by
tracking $\xi^1 \equiv 0$ on $[0,T]$. The second agent, on the
contrary, does not pursue any predetermined buying or selling
objectives but solely chooses to trade in the risky asset because he
knows about the intentions of the first liquidating agent. That is,
possessing no shares at time 0 ($x^2=0$) he gives himself the
constraints $\xi^2_t = \Xi^2_T = 0$ $\PP$-a.s. for all $t \in
[0,T]$. In this case, following Theorem~\ref{thm:main}, we have
$\hat{\xi}^1 \equiv \hat{\xi}^2 \equiv 0$ $\PP$-a.s. on $[0,T]$
in~\eqref{def:optsignal1} and~\eqref{def:optsignal2}, and the
deterministic equilibrium trading rates of both players
in~\eqref{eq:Xsol} reduce to
\begin{equation} \label{eq:ratesexample1}
  \hat{\alpha}^1_t = 
  \frac{c^+_t+c^-_t}{2\lambda} \left( - w^5_t
    \hat{X}^2_t - \hat{X}^1_t \right) \quad \text{and} \quad 
  \hat{\alpha}^2_t = 
  \frac{c^+_t+c^-_t}{2\lambda} \left( - w^5_t
    \hat{X}^1_t - \hat{X}^2_t \right) 
\end{equation}
on $[0,T)$; cf. also the result in~\cite[Corollary
2.6]{SchiedZhang:17} with a slightly different representation. We
observe in~\eqref{eq:ratesexample1} that the first agent's portfolio
position $\hat{X}^1_t$ is not gradually reverting towards 0 but takes
the effect of the second agent's actions into account via the
correction term $-w^5_t \hat{X}^2_t$. Similarly, concerning the second
agent, it is optimal for him to systematically trade in the direction
of the liquidating agent's current portfolio position~$\hat{X}^1_t$
weighted with $w^5_t \in (-1,1)$.

\begin{figure}
  \begin{center}
    \includegraphics[scale=.41]{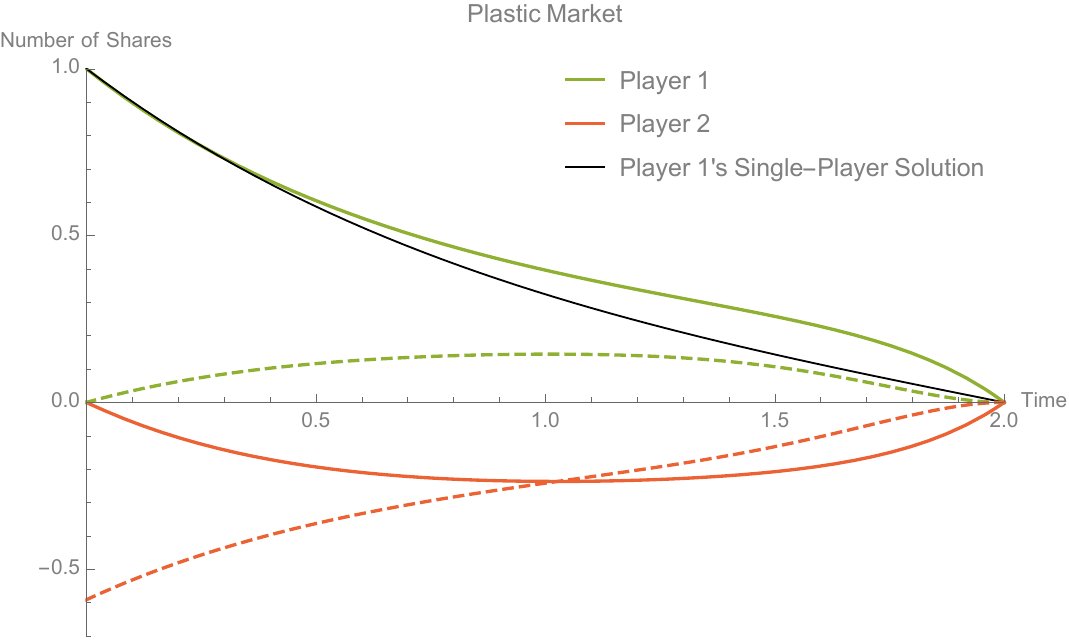}
    \includegraphics[scale=.41]{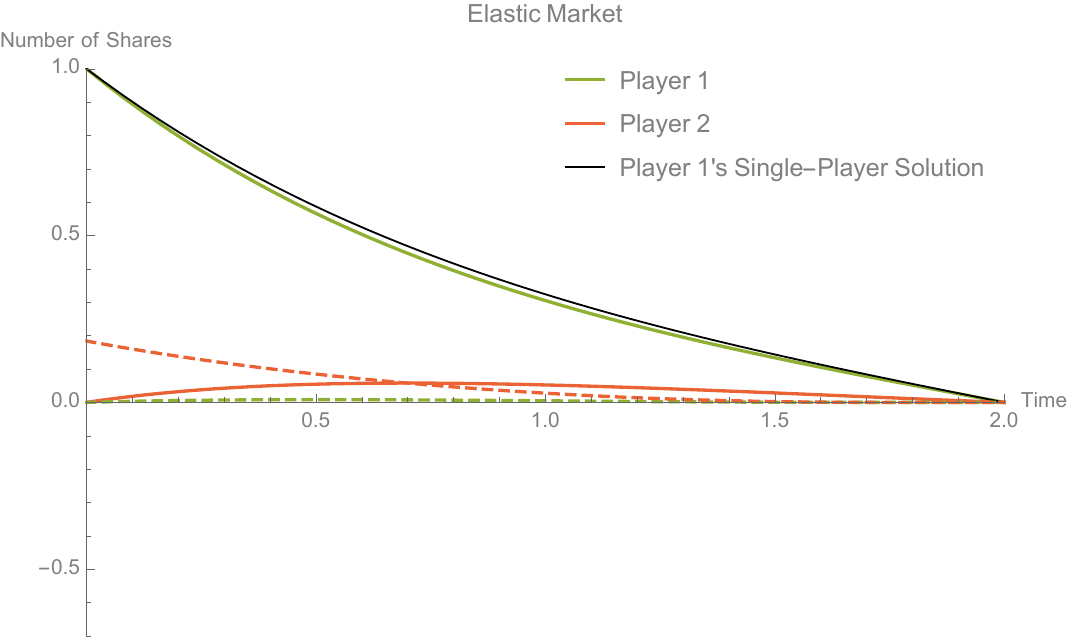}
    \caption{The two-player Nash equilibrium strategies $\hat{X}^1$
      for the liquidating agent 1 (green) and $\hat{X}^2$ for agent 2
      (orange) on $[0,T]$, together with the corresponding processes
      $-w^5 \hat{X}^i$ $(i=1,2)$ from the trading rates
      in~\eqref{eq:ratesexample1} (same-color dashed lines). The
      optimal single-player liquidation strategy
      from~\eqref{eq:singlePlayerSol} is depicted in black.  The
      parameters are $T=2$, $\sigma = 1$, $\lambda = 1$, as well as
      $\gamma = 4$ (left panel), $ = 0.2$ (right panel).}
    \label{fig:ex1fig1}
  \end{center}
\end{figure}

As shown in Figure~\ref{fig:ex1fig1}, this yields to predation on the
first agent in a plastic market where, e.g.,
$\gamma = 4 > 1 = \lambda$. Indeed, during the first half of the
trading period he short-sells the risky asset in parallel to the
selling of the first agent and then steadily unwinds his accrued short
position by buying back shares to become ``hands-clean'' by final
time~$T$. In contrast, in an elastic market with, e.g.,
$\gamma = 0.2 < 1 = \lambda$, the Nash equilibrium strategy dictates
the second agent to cooperate with the seller and to moderately buy
almost up to one-tenth of the shares by time $T/2$ agent~1 is
concurrently selling before starting liquidating his portfolio to
finish up with zero inventory at $T$. Note that the weight function
$w^5_\cdot$ in~\eqref{eq:ratesexample1} flips sign depending on the
market's illiquidity regime (see also Figure~\ref{fig:weights}). As a consequence,
compared to the single-player optimal liquidation strategy
$\hat{X}_t = 1 + \int_0^t \hat{\alpha}_s ds$, $t \in [0,T]$, which
satisfies
\begin{equation} \label{eq:singlePlayerSol}
  \hat{\alpha}_t =  -\sqrt{\frac{\sigma}{\lambda}}  \coth\left(
    \sqrt{\frac{\sigma}{\lambda}} (T-t) \right) \hat{X}_t \qquad (0
  \leq t < T)
\end{equation}
(cf., e.g.,~\citet{Almgren:12}), and does not depend on $\gamma$, we
observe in Figure~\ref{fig:ex1fig1} that, due to the presence of the
second agent's trading activity which directly feeds into the first
agent's turnover rate~$\hat{\alpha}^1$ via $-w^5 \hat{X}^2$
in~\eqref{eq:ratesexample1}, her optimal portfolio liquidation
strategy becomes more prudent in a plastic market and slightly more
aggressive in an elastic market environment. To sum up, in
equilibrium, depending on the illiquid market type, either predation
or cooperation between both agents occurs; see also the discussion
in~\cite[Section 3]{SchiedZhang:17}.

\subsection{Piecewise constant inventory targets}

The next two case studies are again simple deterministic examples but
this time with nonzero optimal signal processes $\hat{\xi}^1$ and
$\hat{\xi}^2$.

In the first example, as in the optimal liquidation problem above, we
suppose that agent~2 only trades in the risky asset because of his
awareness of the trading activity of the first agent. That is, with
$x^2 = 0$ initial shares his inventory targets are
$ \xi^2_t = \Xi^2_T = 0$ $\PP$-a.s. for all $t \in [0,T]$. Concerning
the first agent, starting with no inventory $x^1=0$ she wants to
follow a stock-buying schedule over a time period of $T=10$ that
prescribes to hold one share until time $T/2$ and then to double and
hold her position up to time~$T$. Her inventory target is thus
$\xi^1_t = 1 \cdot 1_{\{0 \leq t < 5\}}+2 \cdot 1_{\{5 \leq t \leq
  10\}}$ on $[0,T]$ with terminal constraint $\Xi^1_T = 2$. Note that
in this game setup the optimal signal processes $\hat{\xi}^1$ and
$\hat{\xi}^2$ of both agents in~\eqref{def:optsignal1}
and~\eqref{def:optsignal2} in equilibrium are nonzero. In particular,
similar to the single-player case in~\cite{BankSonerVoss:17} they are
anticipating and smoothing out the jump in $\xi^1$ at time $T/2$ via
the averaging through the kernels $K^1$ and~$K^2$.
\begin{figure}[!ht]
  \begin{center}
    \includegraphics[scale=.4]{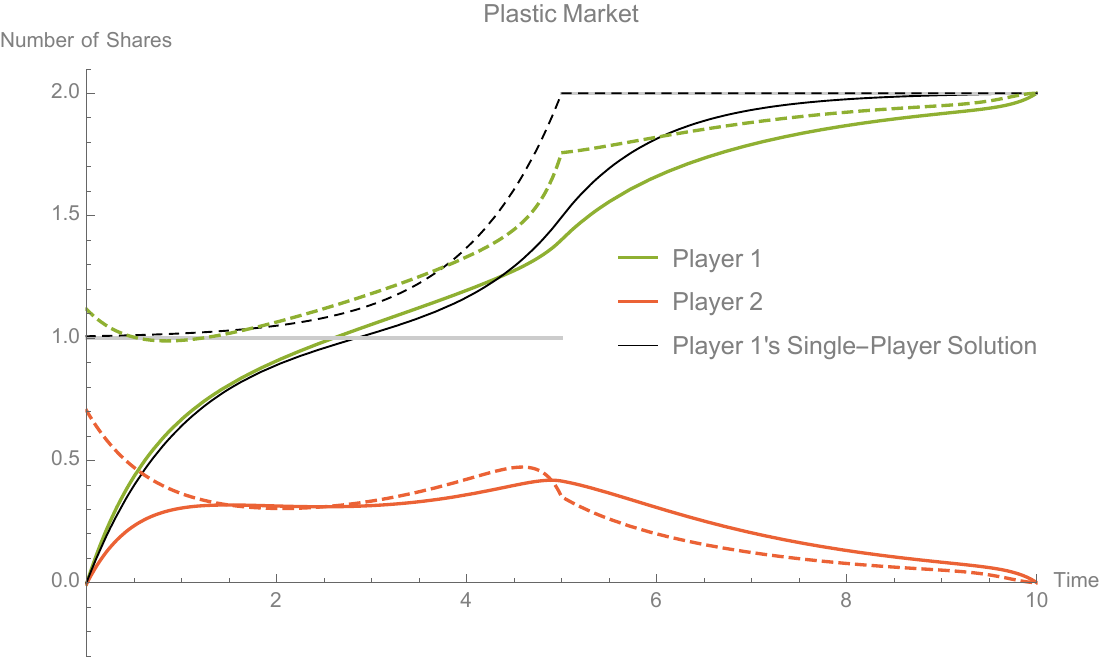}
    \includegraphics[scale=.4]{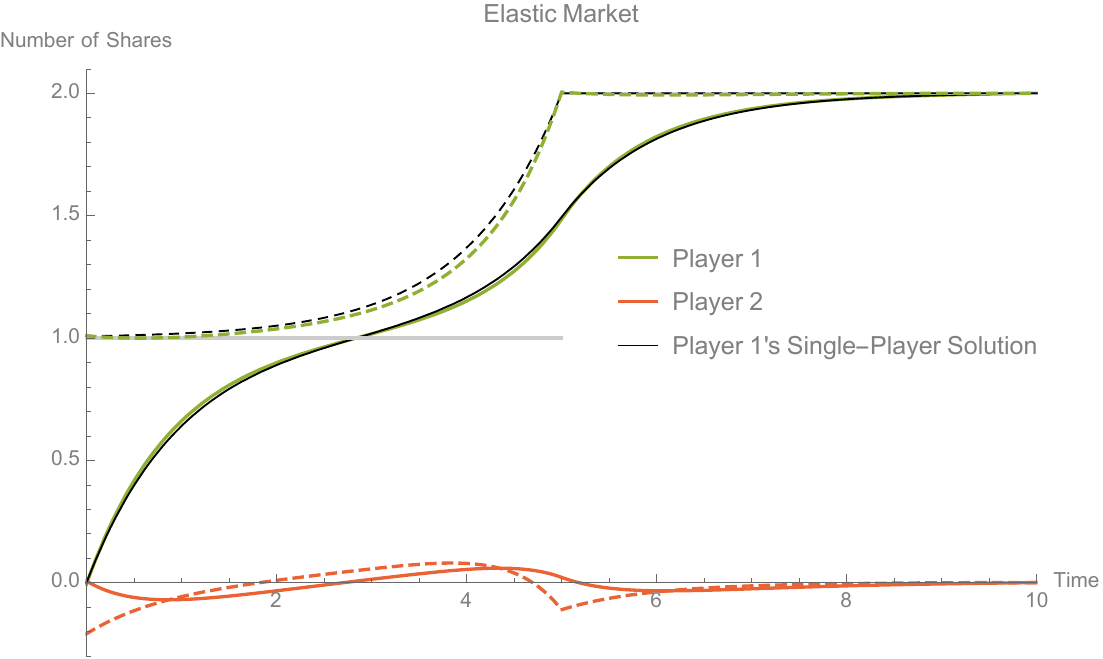}
    \caption{The two-player Nash equilibrium strategies $\hat{X}^1$
      for Player 1 (green) and $\hat{X}^2$ for Player 2 (orange),
      together with the processes $\hat{\xi}^i-w^5 \hat{X}^j$
      $(i\neq j \in \{1,2\})$ from the optimal trading rates
      in~\eqref{eq:Xsol} (same-color dashed lines). The first agent's
      buying program $\xi^1 = 1_{[0,5)}+2 \cdot 1_{[5,10]}$ is
      plotted in grey. For comparison, the corresponding single-player
      optimal tracking strategy with associated optimal signal process
      from~\cite{BankSonerVoss:17} is depicted in black (solid and
      dashed).  The parameters are $T=10$, $\sigma = 1$,
      $\lambda = 1$, as well as $\gamma = 4$ (left panel),
      $\gamma = 0.2$ (right panel).}
    \label{fig:ex2fig1}
  \end{center}
\end{figure}
The associated Nash-equilibrium trading strategies $\hat{X}^1$ and
$\hat{X}^2$ from Theorem~\ref{thm:main} are presented in
Figure~\ref{fig:ex2fig1}. As expected from the liquidation problem
above, if the market is plastic $(\gamma > \lambda)$ the second agent
heavily preys on the first agent by trading halfway of the trading
period in the same direction and buying shares. Accordingly, in
comparison to the first agent's single-player optimal tracking
strategy from~\cite{BankSonerVoss:17} (which does not dependent on
$\gamma$) her running after the buying-schedule $\xi^1$ gets affected
due to the presence of the preying second agent and falls
behind the single-player solution in the second half of the trading period (also recall the
adjustment $\hat{\xi}^1 -w^5\hat{X}^2$ of the first agent's optimal
signal process in her trading rate in~\eqref{eq:Xsol}).  However, if
the market is elastic $(\lambda > \gamma)$ the second agent's optimal
behaviour in equilibrium changes. Interestingly, we observe that his
strategy turns out to be a succession of round-trips during which he
either provides liquidity to his opponent by short-selling the risky
asset like, e.g., during the first quarter of the trading period, or
engages in predatory trading by concurrently building up some
inventory in parallel to his adversary's buying efforts as it is the
case during the second quarter of the trading period. Thus, compared
to the first agent's single-player optimal strategy, she suitably buys
slightly faster and slower in the two-player setup. Overall, it turns
out that predation and cooperation coexist in equilibrium in this
case.

As a second example, let us examine the situation where both agents
with zero initial inventory $x^1=x^2=0$ seek to gradually build up and
hold a positive fraction of the risky asset over some time period
$[0,T]$ with~$T=10$. Concretely, assume that
$\xi^1 \equiv \Xi^1_T = 1$ and $\xi^2 \equiv \Xi^2_T = 0.1$, i.e.,
agent 1 wants her inventory to be close to 1 and ten times larger than
the desired inventory level of agent~2 all through the trading period
$[0,T]$.
\begin{figure}[!ht]
  \begin{center}
    \includegraphics[scale=.41]{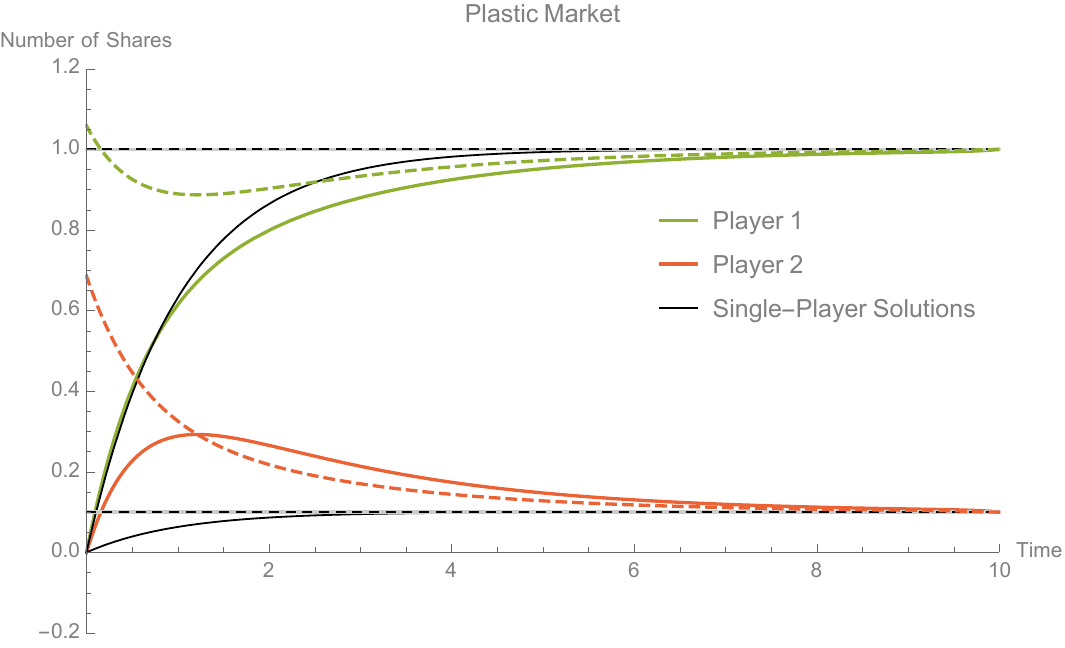}
    \includegraphics[scale=.41]{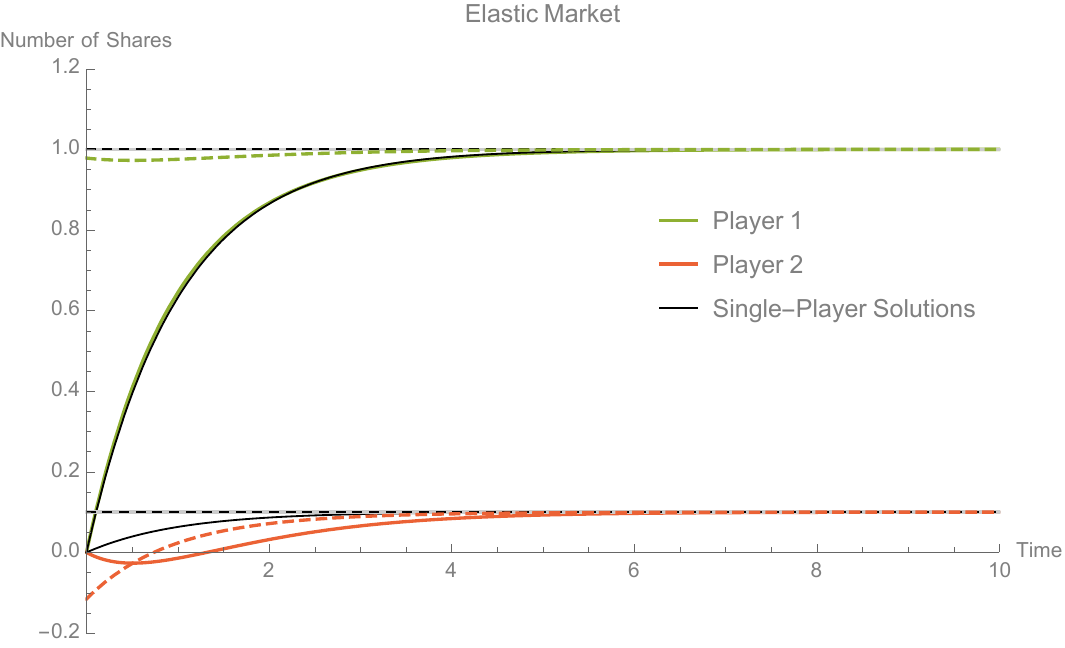}
    \caption{The two-player Nash equilibrium strategies $\hat{X}^1$
      for Player 1 (green) and $\hat{X}^2$ for Player 2 (orange),
      together with the processes $\hat{\xi}^i-w^5 \hat{X}^j$
      $(i\neq j \in \{1,2\})$ from the optimal trading rates
      in~\eqref{eq:Xsol} (same-color dashed lines). Both agent's
      inventory targets $\xi^1 \equiv 1$ and $\xi^2 \equiv 0.1$ are
      plotted in grey. For comparison, the corresponding single-player
      optimal tracking strategies with associated optimal signal
      processes from~\cite{BankSonerVoss:17} are depicted in black
      (solid and dashed). The parameters are $T=10$, $\sigma = 1$,
      $\lambda = 1$, as well as $\gamma = 4$ (left panel),
      $\gamma = 0.2$ (right panel).}
    \label{fig:ex2fig2}
  \end{center}
\end{figure}
The associated Nash equilibrium strategies $\hat{X}^1$ and $\hat{X}^2$
from Theorem~\ref{thm:main} are presented in Figure~\ref{fig:ex2fig2}.
Again, as expected from the analysis above, in a plastic market it is
optimal for agent 2 to excessively prey on the first agent who aims
for a much larger asset position by buying up to three times more
shares than his actual target inventory predetermines. In response,
the acquisition of the first agent is slowed down compared to her
single-player optimal strategy from~\cite{BankSonerVoss:17}. By
contrast, in an elastic market environment it turns out to be optimal
for the second agent to initially ignore her own tracking target and
to trade \emph{away} from her desired inventory level in order to
provide liquidity to the higher-volume seeking first agent by
short-selling some shares. Also note how in this case the second
agent's single-player optimal tracking strategy
from~\cite{BankSonerVoss:17} strongly differs from her optimal
behaviour in the two-player Nash equilibrium at the beginning of the
trading period.

\subsection{Running after the delta}

In the final two examples we want to investigate a situation where the
target strategies $\xi^1$ and $\xi^2$ are adapted stochastic
processes. Specifically, let us suppose that the first agent wants to
hedge an at-the-money call option with maturity $T$ on the underlying
unaffected price process $P=P_0+\sqrt{\sigma} W$ in~\eqref{def:exPrice} by tracking the
corresponding frictionless (Bachelier-)delta-hedging strategy
\begin{equation} \label{def:delta}
  \xi^1_t \set \Phi \left( \frac{P_t - P_0}{\sqrt{\sigma(T-t)}} \right)
  \quad (0 \leq t \leq T).
\end{equation}
Here, $\Phi$ denotes the cumulative distribution function of the
standard normal distribution. We further suppose that her initial
position in the risky asset coincides with the frictionless delta
$x^1=\xi^1_0 = 1/2$ and that $\Xi^1_T = 0$ $\PP$-a.s., i.e., she wants
to systematically unwind her hedging portfolio when approaching
maturity~$T$.

\begin{Lemma}
  The process $(\xi^1_t)_{0 \le t \leq T}$ in~\eqref{def:delta} is a martingale on~$[0,T]$.  
\end{Lemma}

\begin{proof}
  Obviously, $(\xi^1_t)_{0 \le t \leq T}$ is adapted, bounded and hence integrable. Moreover, using the property that for any $a,b \in \mathbb{R}$ a standard normal distributed random variable~$Z$ satisfies $\mathbb{E}[\Phi(a Z + b)] = \Phi(b/\sqrt{1+a^2})$ we obtain
  \begin{equation*}
\mathbb{E}\left[ \Phi \left( \frac{P_t - P_0}{\sqrt{\sigma(T-t)}} \right) \bigg\vert \, \mathcal{F}_s \right] = \mathbb{E}\left[ \Phi \left( \frac{\sqrt{\sigma (t-s)} Z + P_s - P_0}{\sqrt{\sigma(T-t)}} \right)\right] = \Phi \left( \frac{P_s - P_0}{\sqrt{\sigma(T-s)}} \right)
  \end{equation*}
  as desired.
\end{proof}

Firstly, we assume that the second agent does not pursue any specific
predetermined trading objectives, that is, $x^2 = \xi^2 = \Xi^2_T = 0$
$\PP$-a.s. Since~$\xi^1$ in~\eqref{def:delta} is a martingale
on~$[0,T]$ the optimal signal processes $\hat{\xi}^1$ and
$\hat{\xi}^2$ in~\eqref{def:optsignal1} and~\eqref{def:optsignal2}
simplify to
\begin{equation} \label{ex3:signals}
  \hat{\xi}^1_t = (w^3_t + w^4_t) \xi^1_t \quad \text{and} \quad
  \hat{\xi}^2_t = (w^3_t - w^4_t) \xi^1_t \qquad (0 \leq t \leq T),
\end{equation}
using Fubini's theorem and the fact that for each $t \in [0,T)$ the kernels $K^1(t,u)$ and $K^2(t,u)$ as functions in $u \in [t,T)$ integrate to one over $[t,T]$. The Nash equilibrium strategies $\hat{X}^1$ and $\hat{X}^2$ from
Theorem~\ref{thm:main} are plotted in Figure~\ref{fig:ex3fig1},
together with the corresponding realisation of the delta-hedge~$\xi^1$
in the case where the call option expires in the money.
\begin{figure}[!ht] 
  \begin{center}
    \includegraphics[scale=.41]{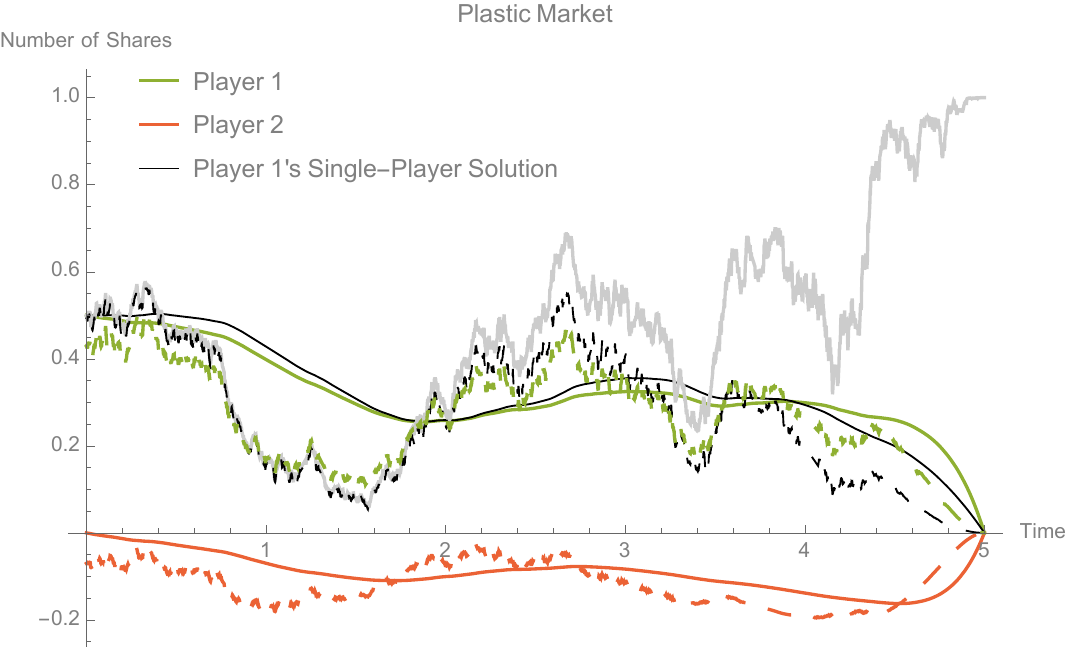}
    \includegraphics[scale=.41]{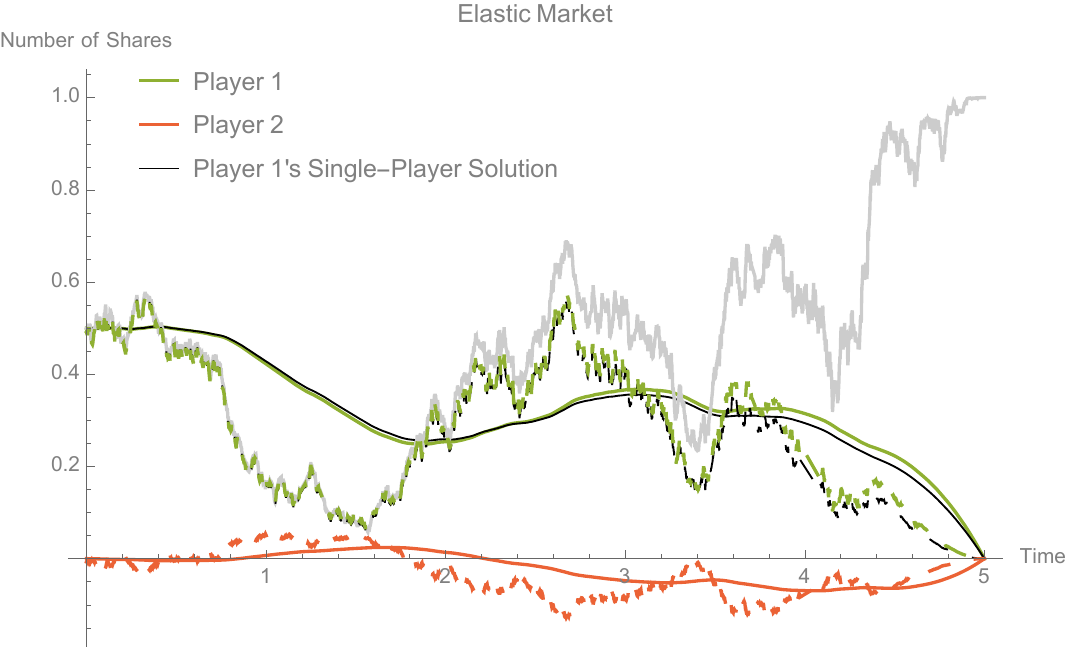}
    \caption{The two-player Nash equilibrium strategies $\hat{X}^1$
      for Player 1 (green) and $\hat{X}^2$ for Player 2 (orange),
      together with the processes $\hat{\xi}^i-w^5 \hat{X}^j$
      $(i\neq j \in \{1,2\})$ from the optimal trading rates
      in~\eqref{eq:Xsol} (same-color dashed lines). The first agent's
      frictionless delta-hedge $\xi^1$ is plotted in grey. For
      comparison, her corresponding single-player optimal hedging
      strategy with associated optimal signal process
      from~\cite{BankSonerVoss:17} is depicted in black (solid and
      dashed). The parameters are $T=5$, $\sigma = 1$, $\lambda = 1$,
      as well as $\gamma = 4$ (upper panel), $\gamma = 0.2$ (lower
      panel).}
    \label{fig:ex3fig1}
  \end{center}
\end{figure}
Depending on the illiquidity parameters, we observe the same
behavioral patterns in equilibrium as in the deterministic cases
analyzed above: In a plastic market environment, the second agent
engages in predatory trading on the first agent by trading in parallel
in the same direction of the delta-hedge. When the market is elastic
he turns into a liquidity provider instead and partially takes the
opposite side of the hedger's transactions. Also note that the sign of
the second agent's optimal signal process in~\eqref{ex3:signals} is
determined by the relation between the weights $w^3$ and $w^4$, which
is in turn affected by the relation between~$\gamma$ and~$\lambda$
(cf. also Figure~\ref{fig:weights}).

\begin{figure}[!ht]
  \begin{center}
    \includegraphics[scale=.39]{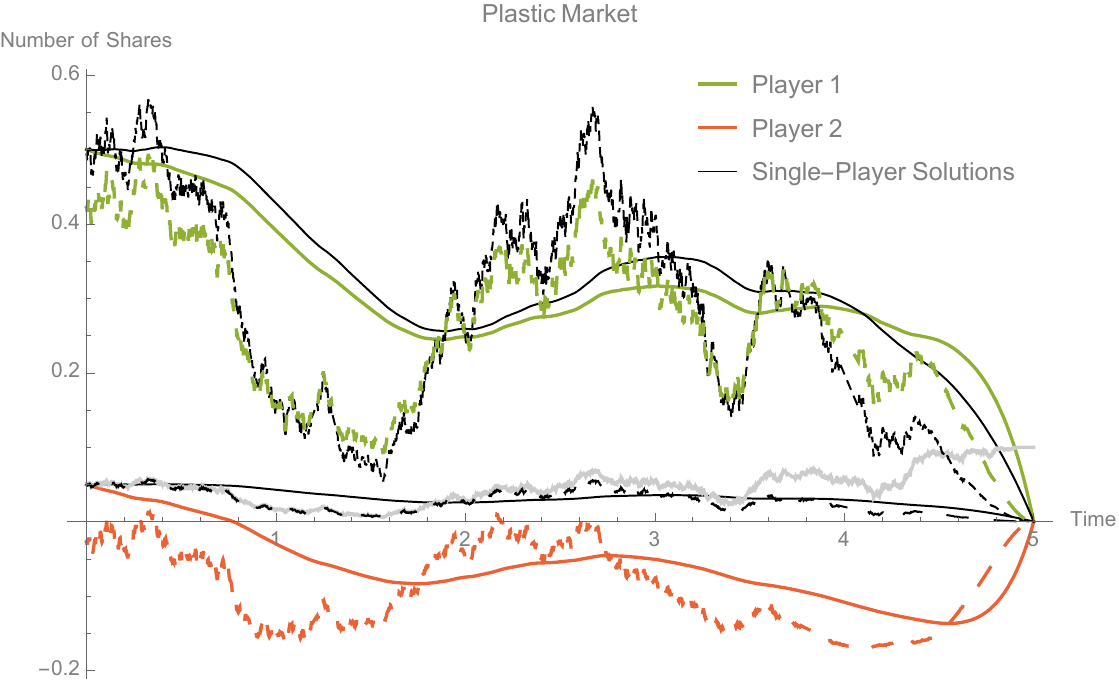}
    \includegraphics[scale=.39]{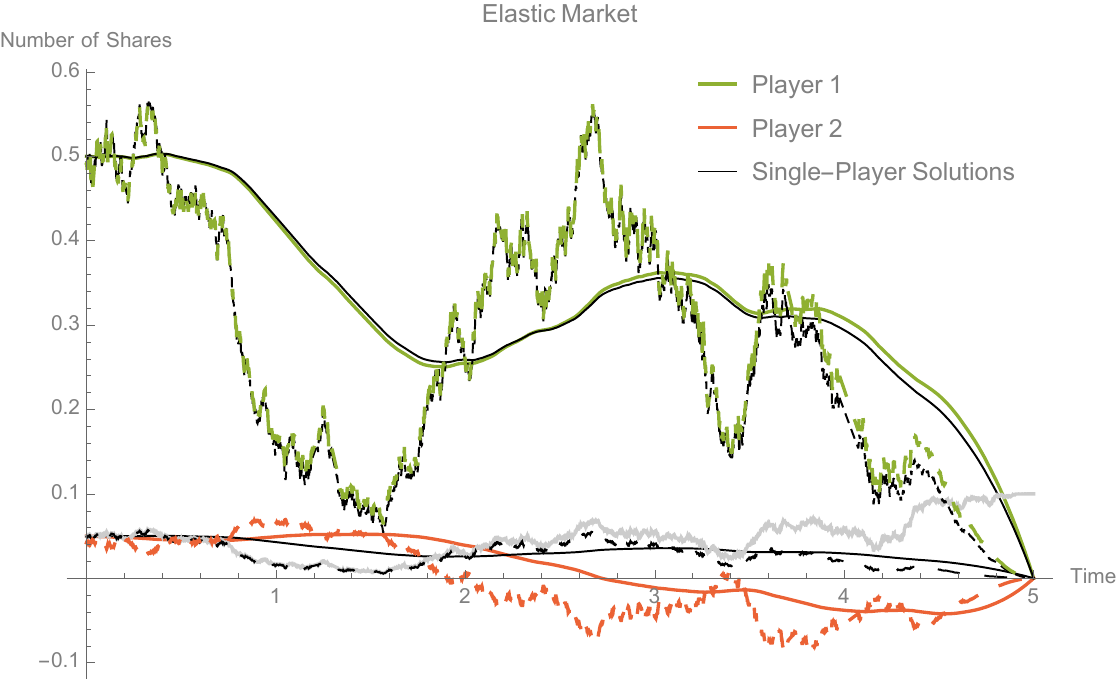}
    \caption{The two-player Nash equilibrium strategies $\hat{X}^1$
      for Player 1 (green) and $\hat{X}^2$ for Player 2 (orange),
      together with the processes $\hat{\xi}^i-w^5 \hat{X}^j$
      $(i\neq j \in \{1,2\})$ from the optimal trading rates
      in~\eqref{eq:Xsol} (same-color dashed lines). Only the second
      agent's frictionless delta-hedge $\xi^2=\xi^1/10$ is plotted in
      grey (the first agent's target strategy $\xi^1$ is the same as
      in Figure~\ref{fig:ex3fig1} and omitted here). For comparison,
      the corresponding single-player optimal hedging strategies of
      the two agents together with their associated optimal signal
      processes from~\cite{BankSonerVoss:17} are depicted in black
      (solid and dashed). The parameters are $T=5$, $\sigma = 1$,
      $\lambda = 1$, as well as $\gamma = 4$ (upper panel),
      $\gamma = 0.2$ (lower panel).}
    \label{fig:ex3fig2}
  \end{center}
\end{figure}

Secondly, let us now assume that the second agent also hedges a
one-tenth fraction of the same call option, i.e., $\xi^2 = \xi^1/10$
(with initial and final portfolio positions $x^2=1/20$ and
$\Xi^2_T =0$ $\PP$-a.s.). The resulting Nash equilibrium strategies
from Theorem~\ref{thm:main} are presented in Figure~\ref{fig:ex3fig2}
where we used the same realisation of the delta-hedge as in
Figure~\ref{fig:ex3fig1}. In a similar vein as in the deterministic
case above, the second agent's optimal behaviour in the two-player
Nash equilibrium changes notably compared to his optimal single-player
frictional hedging strategy from~\cite{BankSonerVoss:17}; focussing
more on preying on the first agent's larger hedging portfolio in a
plastic market, or on providing liquidity to the latter in an elastic
market.

\begin{figure}[!ht]
  \begin{center}
    \includegraphics[scale=.39]{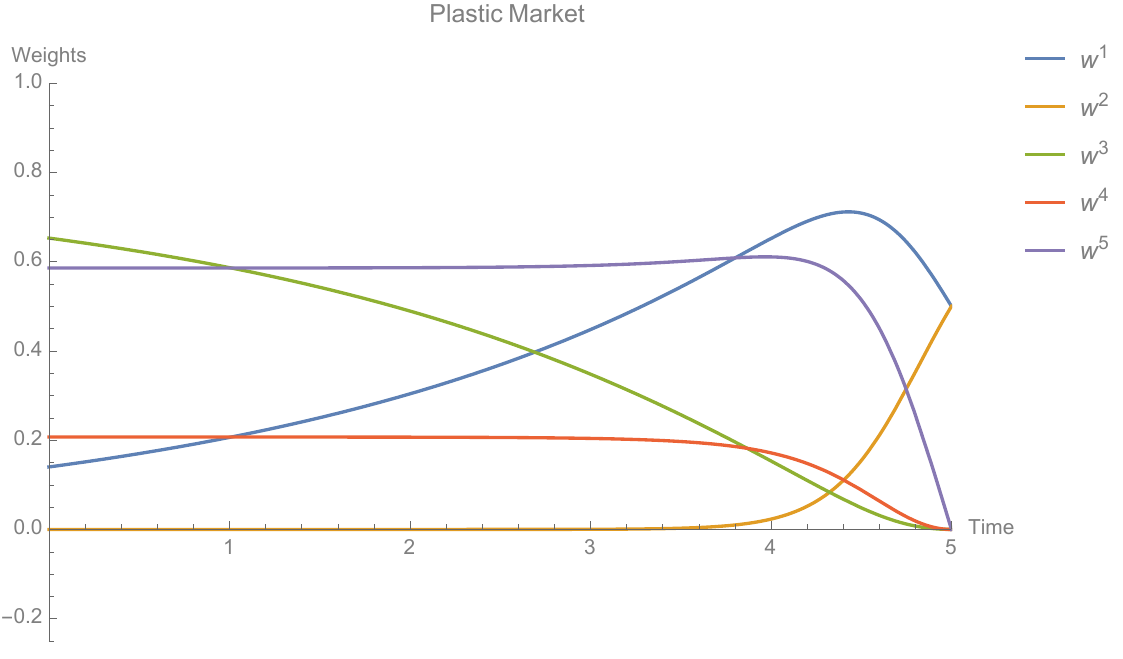}
    \includegraphics[scale=.39]{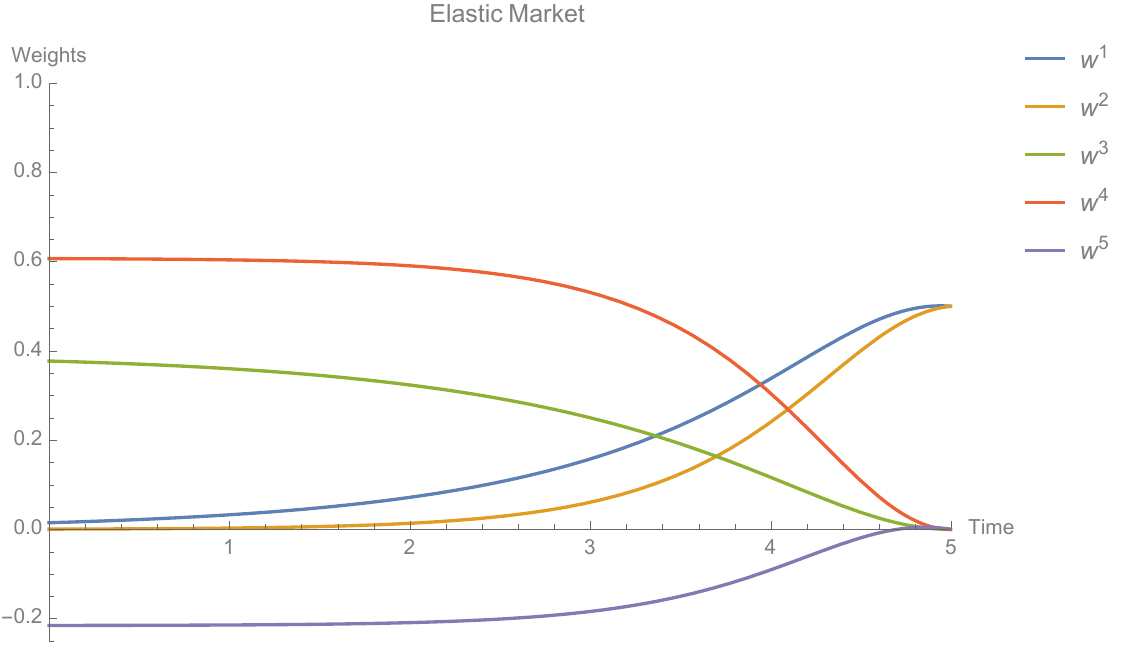}
    \caption{Exemplary illustration of the weight functions $w^1$,
      $w^2$, $w^3$, $w^4$, $w^5$ on $[0,T]$ defined
      in~\eqref{def:weights}. The parameters are $T=5$, $\sigma = 1$,
      $\lambda = 1$, as well as $\gamma = 4$ (upper panel),
      $\gamma = 0.2$ (lower panel).}
    \label{fig:weights}
  \end{center}
\end{figure}

\section*{Appendix} \label{sec:appendicx}
Since the proof of Theorem~\ref{thm:main} is a verification of a proposed Nash equilibrium, we briefly explain for the reader's convenience how the candidate Nash equilibrium strategies $(\hat{\alpha}^1,\hat{\alpha}^2)$ provided in~\eqref{eq:Xsol} can be constructed. Suppose we replace the constrained optimization problems in~\eqref{eq:optProbAgent1} and~\eqref{eq:optProbAgent2} by their unconstrained versions  
\begin{align} 
     J^{1,n}(\alpha^1;\alpha^2) \triangleq & \, J^1(\alpha^1;\alpha^2) + \frac{n}{2} \, \mathbb{E}[(X^1_T - \Xi^1_T)^2] \rightarrow \min_{\alpha^1 \in \cA},
     \label{def:auxprob1} \\
     J^{2,n}(\alpha^2;\alpha^1) \triangleq & \, J^2(\alpha^2;\alpha^1) + \frac{n}{2} \, \mathbb{E}[(X^2_T - \Xi^2_T)^2] \rightarrow \min_{\alpha^2 \in \cA} \label{def:auxprob2}
\end{align}
with some penalty parameter $n \in \mathbb{N}$. Then, along the same lines of Lemmas~\ref{lem:convex},~\ref{lem:uniqueNE}, \ref{lem:gateaux}, \ref{lem:FOC} above, solving~\eqref{def:auxprob1} and~\eqref{def:auxprob2} simultaneously results into solving following coupled FBSDE system
  \begin{equation} \label{eq:auxFBSDE}
    \left\{
    \begin{aligned}
      dX^1_t = & \; \alpha^1_t dt, \qquad X^1_0 = x^1, \\
      dX^2_t = & \; \alpha^2_t dt, \qquad X^2_0 = x^2, \\
      d\alpha^1_t = & \, \frac{\sigma}{\lambda} (X^1_t - \xi^1_t) dt -
      \frac{\gamma}{2\lambda} \alpha^2_t dt - \frac{1}{2} d\alpha^2_t +
      d\tilde{M}^1_t, \\
      \alpha^1_T = & \, -\frac{n}{\lambda} (X^1_T - \Xi^1_T) - \frac{1}{2} \alpha^2_T -\frac{\gamma}{2\lambda} (X^2_T - x^2), \\
      d\alpha^2_t = & \; \frac{\sigma}{\lambda} (X^2_t - \xi^2_t) dt -
      \frac{\gamma}{2\lambda} \alpha^1_t dt - \frac{1}{2} d\alpha^1_t +
      d\tilde{M}^2_t, \\
      \alpha^2_T = & \, -\frac{n}{\lambda} (X^2_T - \Xi^2_T) - \frac{1}{2} \alpha^1_T -\frac{\gamma}{2\lambda} (X^1_T - x^1)
    \end{aligned}
    \right.
  \end{equation}
for two suitable square integrable martingales $(\tilde{M}^1_t)_{0 \leq t \leq T}$ and $(\tilde{M}^2_t)_{0 \leq t \leq T}$. The system in~\eqref{eq:auxFBSDE} can be decoupled by adding and subtracting both forward and backward equations to obtain the two autonomous systems
\begin{equation} \label{eq:auxFBSDEplus}
    \left\{
    \begin{aligned}
      d(X^1_t+X^2_t) = & \; (\alpha^1_t + \alpha^2_t) dt, \qquad X^1_0 + X^2_0 = x^1 + x^2, \\
      d(\alpha^1_t + \alpha^2_t) = & \, \frac{2\sigma}{3\lambda} \left( (X^1_t+X^2_t) - (\xi^1_t+\xi^2_t) \right) dt -
      \frac{\gamma}{3\lambda} (\alpha^1_t+\alpha^2_t) dt + \frac{2}{3} d(\tilde{M}^1_t + \tilde{M}^2_t), \\
      \alpha^1_T + \alpha^2_T = & \, -\frac{2n}{3\lambda} \left( (X^1_T + X^2_T) - (\Xi^1_T+\Xi^2_T) \right) -\frac{\gamma}{3\lambda} \left( (X^1_T + X^2_T) - (x^1 + x^2) \right),
    \end{aligned}
    \right.
  \end{equation}
and
\begin{equation} \label{eq:auxFBSDEminus}
    \left\{
    \begin{aligned}
      d(X^1_t-X^2_t) = & \; (\alpha^1_t - \alpha^2_t) dt, \qquad X^1_0 - X^2_0 = x^1 - x^2, \\
      d(\alpha^1_t - \alpha^2_t) = & \, \frac{2\sigma}{\lambda} \left( (X^1_t-X^2_t) - (\xi^1_t-\xi^2_t) \right) dt +
      \frac{\gamma}{\lambda} (\alpha^1_t-\alpha^2_t) dt + 2 d(\tilde{M}^1_t - \tilde{M}^2_t), \\
      \alpha^1_T - \alpha^2_T = & \, -\frac{2n}{\lambda} \left( (X^1_T - X^2_T) - (\Xi^1_T+\Xi^2_T) \right) + \frac{\gamma}{\lambda} \left( (X^1_T - X^2_T) - (x^1 - x^2) \right).
    \end{aligned}
    \right.
  \end{equation}
The decoupled FBSDEs in~\eqref{eq:auxFBSDEplus} and~\eqref{eq:auxFBSDEminus} are linear. To solve them, we make a linear ansatz of the following form
\begin{equation} \label{eq:ansatz}
    \lambda (\alpha^1_t + \alpha^2_t) = b^{+,n}_t - c^{+,n}_t (X^1_t+X^2_t),\quad  \lambda (\alpha^1_t - \alpha^2_t) = b^{-,n}_t - c^{-,n}_t (X^1_t-X^2_t). 
\end{equation}
Plugging this ansatz in~\eqref{eq:auxFBSDEplus} and~\eqref{eq:auxFBSDEminus}, respectively, and comparing coefficients yields two deterministic Riccati equations for $c^{+,n}$ and $c^{-,n}$ given by
\begin{equation} \label{eq:riccatiaux} 
\begin{aligned}
(c^{+,n}_t)' = & \, \frac{(c^{+,n}_t)^2}{\lambda} - \frac{\gamma}{3\lambda} c^{+,n}_t - \frac{2}{3} \sigma, \quad c^{+,n}_T = \frac{1}{3} (2 n + \gamma), \\
(c^{-,n}_t)' = & \, \frac{(c^{-,n}_t)^2}{\lambda} +
  \frac{\gamma}{\lambda} c^{-,n}_t - 2\sigma, \quad c^{-,n}_T = (2 n - \gamma);
\end{aligned}
\end{equation}
as well as two linear BSDEs for $b^{+,n}$ and $b^{-,n}$ given by
\begin{equation} \label{eq:BSDEaux} 
\begin{aligned}
db^{+,n}_t = & \, \left( \left( \frac{c^{+,n}_t}{\lambda} - \frac{\gamma}{3\lambda} \right) b^{+,n}_t - \frac{2\sigma}{3} (\xi^1_t + \xi^2_t) \right) dt - \frac{2\lambda}{3} d(\tilde{M}^1_t + \tilde{M}^2_t), \\
b^{+,n}_T = & \, \frac{2n}{3} (\Xi^1_T + \Xi^2_T) + \frac{\gamma}{3} (x^1+x^2), \\
db^{-,n}_t = & \, \left( \left( \frac{c^{-,n}_t}{\lambda} + \frac{\gamma}{\lambda} \right) b^{-,n}_t - 2\sigma (\xi^1_t - \xi^2_t) \right) dt - 2\lambda d(\tilde{M}^1_t - \tilde{M}^2_t), \\
b^{-,n}_T = & \, 2n (\Xi^1_T - \Xi^2_T) - \gamma (x^1-x^2).
\end{aligned}
\end{equation}
The ODEs in~\eqref{eq:riccatiaux} can be solved in closed form with solutions
\begin{equation} \label{eq:ODEauxSol} 
c^{+,n}_t = \frac{1}{6} \gamma + \frac{1}{3} \sqrt{\delta^+} \frac{e^{\frac{2\sqrt{\delta^+}}{3\lambda} (T-t)} \kappa^+_n - 1}{e^{\frac{2\sqrt{\delta^+}}{3\lambda} (T-t)} \kappa^+_n + 1}, \quad c^{-,n}_t = -\frac{\gamma}{2} + \sqrt{\delta^-} \frac{e^{\frac{2\sqrt{\delta^-}}{\lambda} (T-t)} \kappa^-_n - 1}{e^{\frac{2\sqrt{\delta^-}}{\lambda} (T-t)} \kappa^-_n + 1},
\end{equation}
where $\kappa^+_n \triangleq \frac{2\sqrt{\delta^+} + \gamma + 4n}{2\sqrt{\delta^+}-\gamma-4n}$ and $\kappa^-_n \triangleq \frac{2\sqrt{\delta^-} - \gamma + 4n}{2\sqrt{\delta^-}+\gamma-4n}$ (with $\delta^+, \delta^-$ introduced in~\eqref{def:deltaplusminus}). Also the linear BSDEs in~\eqref{eq:BSDEaux} have explicit solutions given by
\begin{equation} \label{eq:BSDEauxSol} 
\begin{aligned}
b^{+,n}_t = & \; \mathbb{E}\left[ \left( \frac{2n}{3} (\Xi^1_T + \Xi^2_T) + \frac{\gamma}{3} (x^1+x^2) \right) e^{-\int_t^T \big( \frac{c^{+,n}_s}{\lambda} - \frac{\gamma}{3\lambda} \big) ds} \right. \\
& \hspace{20pt} + \left. \int_t^T \frac{2\sigma}{3} (\xi^1_s + \xi^2_s) e^{-\int_t^s \big( \frac{c^{+,n}_u}{\lambda} - \frac{\gamma}{3\lambda} \big) du} \, ds \; \bigg\vert \; \mathcal{F}_t \right], \\
b^{-,n}_t = & \; \mathbb{E}\left[ \left( 2n (\Xi^1_T - \Xi^2_T) - \gamma (x^1-x^2) \right) e^{-\int_t^T \big( \frac{c^{-,n}_s}{\lambda} + \frac{\gamma}{\lambda} \big) ds} \right. \\
& \hspace{20pt} + \left. \int_t^T 2\sigma (\xi^1_s - \xi^2_s) e^{-\int_t^s \big( \frac{c^{-,n}_u}{\lambda} + \frac{\gamma}{\lambda} \big) du} \, ds \; \bigg\vert \; \mathcal{F}_t \right].
\end{aligned}
\end{equation}
Putting everything together with the ansatz in~\eqref{eq:ansatz}, we obtain (for every $n \in \mathbb{N}$) a pair $(\alpha^1, \alpha^2)$ of candidate solutions which simultaneously solve~\eqref{def:auxprob1} and~\eqref{def:auxprob2}, namely
\begin{equation} \label{eq:auxPorblemSol}
\begin{aligned}
\alpha^1_t = & \, \frac{1}{2\lambda} \left( \lambda (\alpha^1_t + \alpha^2_t) + \lambda (\alpha^1_t - \alpha^2_t) \right) \\
= & \, \frac{c^{+,n}_t+c^{-,n}_t}{2\lambda} \left( \frac{b^{+,n}_t + b^{-,n}_t}{c^{+,n}_t+c^{-,n}_t} - \frac{c^{+,n}_t-c^{-,n}_t}{c^{+,n}_t+c^{-,n}_t} X^{2}_t - X^{1}_t \right), \\ 
\alpha^{2}_t = & \, \frac{1}{2\lambda} \left( \lambda (\alpha^1_t + \alpha^2_t) - \lambda (\alpha^1_t - \alpha^2_t) \right) \\
= & \, \frac{c^{+,n}_t+c^{-,n}_t}{2\lambda} \left( \frac{b^{+,n}_t - b^{-,n}_t}{c^{+,n}_t+c^{-,n}_t} - \frac{c^{+,n}_t-c^{-,n}_t}{c^{+,n}_t+c^{-,n}_t} X^1_t - X^2_t \right).
\end{aligned}    
\end{equation}
Since all terms in~\eqref{eq:auxPorblemSol} can be explicitly computed, one can identify the limit in~\eqref{eq:auxPorblemSol} as the penalty parameter $n$ in~\eqref{def:auxprob1} and~\eqref{def:auxprob2} goes to infinity. This yields $(\hat{\alpha}^1,\hat{\alpha}^2)$ in~\eqref{eq:Xsol}, a candidate for the Nash equilibrium strategies for the original constraint stochastic differential game from Section~\ref{sec:problem}. It is then only left to show that $(\hat{\alpha}^1,\hat{\alpha}^2)$ is indeed the unique Nash equilibrium and belongs to $\mathcal{A}^1 \times \mathcal{A}^2$. This verification is carried out in the proof of Theorem~\ref{thm:main} in Section~\ref{sec:main} above.

\section*{Acknowledgement}
I am grateful to Jean-Pierre Fouque for encouraging and illuminating discussions. The paper has also profoundly benefited from the valuable comments and suggestions of an anonymous referee and Ulrich Horst.


\bibliographystyle{plainnat} \bibliography{finance}

\begin{thebibliography}{34}
\providecommand{\natexlab}[1]{#1}
\providecommand{\url}[1]{\texttt{#1}}
\expandafter\ifx\csname urlstyle\endcsname\relax
  \providecommand{\doi}[1]{doi: #1}\else
  \providecommand{\doi}{doi: \begingroup \urlstyle{rm}\Url}\fi

\bibitem[Almgren(2012)]{Almgren:12}
Robert Almgren.
\newblock Optimal trading with stochastic liquidity and volatility.
\newblock \emph{SIAM Journal on Financial Mathematics}, 3\penalty0
  (1):\penalty0 163--181, 2012.
\newblock URL \url{https://doi.org/10.1137/090763470}.

\bibitem[Almgren and Chriss(2001)]{AlmgChr:01}
Robert Almgren and Neil Chriss.
\newblock Optimal execution of portfolio transactions.
\newblock \emph{J. Risk}, 3:\penalty0 5--39, 2001.

\bibitem[Almgren and Li(2016)]{AlmgLi:16}
Robert Almgren and Tianhui~Michael Li.
\newblock Option hedging with smooth market impact.
\newblock \emph{Market Microstructure and Liquidity}, 02\penalty0
  (01):\penalty0 1650002, 2016.
\newblock URL \url{https://doi.org/10.1142/S2382626616500027}.

\bibitem[Attari et~al.(2005)Attari, Mello, and Ruckes]{AttariMelloRuckes:05}
Mukarram Attari, Antonio~S. Mello, and Martin~E. Ruckes.
\newblock Arbitraging arbitrageurs.
\newblock \emph{The Journal of Finance}, 60\penalty0 (5):\penalty0 2471--2511,
  2005.
\newblock ISSN 00221082, 15406261.
\newblock URL \url{http://www.jstor.org/stable/3694755}.

\bibitem[Bank et~al.(2017)Bank, Soner, and Vo{\ss}]{BankSonerVoss:17}
Peter Bank, H.~Mete Soner, and Moritz Vo{\ss}.
\newblock Hedging with temporary price impact.
\newblock \emph{Mathematics and Financial Economics}, 11\penalty0 (2):\penalty0
  215--239, 2017.
\newblock ISSN 1862-9660.
\newblock URL \url{http://dx.doi.org/10.1007/s11579-016-0178-4}.

\bibitem[Brunnermeier and Pedersen(2005)]{BrunnermeierPedersen:05}
Markus~K. Brunnermeier and Lasse~Heje Pedersen.
\newblock Predatory trading.
\newblock \emph{The Journal of Finance}, 60\penalty0 (4):\penalty0 1825--1863,
  2005.
\newblock URL
  \url{https://onlinelibrary.wiley.com/doi/abs/10.1111/j.1540-6261.2005.00781.x}.

\bibitem[Cai et~al.(2017)Cai, Rosenbaum, and Tankov]{CaiRosenbaumTankov:17}
Jiatu Cai, Mathieu Rosenbaum, and Peter Tankov.
\newblock Asymptotic lower bounds for optimal tracking: A linear programming
  approach.
\newblock \emph{Ann. Appl. Probab.}, 27\penalty0 (4):\penalty0 2455--2514, 08
  2017.
\newblock URL \url{https://doi.org/10.1214/16-AAP1264}.

\bibitem[Cardaliaguet and Lehalle(2018)]{CardaliaguetLehalle:18}
Pierre Cardaliaguet and Charles-Albert Lehalle.
\newblock Mean field game of controls and an application to trade crowding.
\newblock \emph{Mathematics and Financial Economics}, 12\penalty0 (3):\penalty0
  335--363, Jun 2018.
\newblock ISSN 1862-9660.
\newblock URL \url{https://doi.org/10.1007/s11579-017-0206-z}.

\bibitem[Carlin et~al.(2007)Carlin, Lobo, and
  Viswanathan]{CarlinLoboViswanathan:07}
Bruce~Ian Carlin, Miguel~Sousa Lobo, and S.~Viswanathan.
\newblock Episodic liquidity crises: Cooperative and predatory trading.
\newblock \emph{The Journal of Finance}, 62\penalty0 (5):\penalty0 2235--2274,
  2007.
\newblock URL
  \url{https://onlinelibrary.wiley.com/doi/abs/10.1111/j.1540-6261.2007.01274.x}.

\bibitem[Carmona and Yang(2008)]{CarmonaYang:09}
Ren\'e Carmona and Joseph Yang.
\newblock Predatory {T}rading: a {G}ame on {V}olatility and {L}iquidity, 2008.
\newblock Preprint, available online at
  \url{https://carmona.princeton.edu/download/fe/PredatoryTradingGameQF.pdf}.

\bibitem[Cartea and Jaimungal(2016)]{CarteaJaimungal16}
Álvaro Cartea and Sebastian Jaimungal.
\newblock A closed-form execution strategy to target volume weighted average
  price.
\newblock \emph{SIAM Journal on Financial Mathematics}, 7\penalty0
  (1):\penalty0 760--785, 2016.
\newblock URL \url{https://doi.org/10.1137/16M1058406}.

\bibitem[Casgrain and Jaimungal(2018)]{CasgrainJaimungal:19}
Philippe Casgrain and Sebastian Jaimungal.
\newblock {M}ean {F}ield {G}ames with {P}artial {I}nformation for {A}lgorithmic
  {T}rading, 2018.
\newblock Preprint, available online at \url{https://arxiv.org/abs/1803.04094}.

\bibitem[Casgrain and Jaimungal(2020)]{CasgrainJaimungal:18}
Philippe Casgrain and Sebastian Jaimungal.
\newblock Mean-field games with differing beliefs for algorithmic trading.
\newblock \emph{Mathematical Finance}, 30\penalty0 (3):\penalty0 995--1034,
  2020.
\newblock URL \url{https://onlinelibrary.wiley.com/doi/abs/10.1111/mafi.12237}.

\bibitem[Chu et~al.(2009)Chu, Lehnert, and Passmore]{ChuLehnertPassmore:09}
Chenghuan~Sean Chu, Andreas Lehnert, and Wayne Passmore.
\newblock Strategic trading in multiple assets and the effects on market
  volatility.
\newblock \emph{International Journal of Central Banking}, 5\penalty0
  (4):\penalty0 143--172, 2009.

\bibitem[Drapeau et~al.(2021)Drapeau, Luo, Schied, and Xiong]{DrapeauSchied:20}
Samuel Drapeau, Peng Luo, Alexander Schied, and Dewen Xiong.
\newblock An {FBSDE} approach to market impact games with stochastic
  parameters.
\newblock \emph{Probability, Uncertainty and Quantitative Risk}, 6\penalty0
  (3):\penalty0 237--260, 2021.

\bibitem[Ekeland and T\'emam(1999)]{EkelTem:99}
Ivar Ekeland and Roger T\'emam.
\newblock \emph{Convex Analysis and Variational Problems}.
\newblock Society for Industrial and Applied Mathematics, 1999.
\newblock URL \url{http://epubs.siam.org/doi/abs/10.1137/1.9781611971088}.

\bibitem[Ekren and Nadtochiy(2022)]{EkrenNadtochiy:19}
Ibrahim Ekren and Sergey Nadtochiy.
\newblock Utility-based pricing and hedging of contingent claims in
  {A}lmgren-{C}hriss model with temporary price impact.
\newblock \emph{Mathematical Finance}, 32\penalty0 (1):\penalty0 172--225,
  2022.
\newblock URL \url{https://onlinelibrary.wiley.com/doi/abs/10.1111/mafi.12330}.

\bibitem[Evangelista and Thamsten(2020)]{EvangelistaThamsten:20}
David Evangelista and Yuri Thamsten.
\newblock On finite population games of optimal trading, 2020.
\newblock Preprint, available online at \url{https://arxiv.org/abs/2004.00790}.

\bibitem[Fu and Horst(2020)]{FuHorst:20}
Guanxing Fu and Ulrich Horst.
\newblock Mean-field leader-follower games with terminal state constraint.
\newblock \emph{SIAM Journal on Control and Optimization}, 58\penalty0
  (4):\penalty0 2078--2113, 2020.
\newblock URL \url{https://doi.org/10.1137/19M1241878}.

\bibitem[Fu et~al.(2020)Fu, Horst, and Xia]{FuHorstXia:20}
Guanxing Fu, Ulrich Horst, and Xiaonyu Xia.
\newblock Portfolio liquidation games with self-exciting order flow, 2020.
\newblock Preprint, available online at \url{https://arxiv.org/abs/2011.05589}.

\bibitem[Fu et~al.(2021)Fu, Graewe, Horst, and Popier]{FuGraeweHorstPopier:20}
Guanxing Fu, Paulwin Graewe, Ulrich Horst, and Alexandre Popier.
\newblock A mean field game of optimal portfolio liquidation.
\newblock \emph{Mathematics of Operations Research}, 46\penalty0 (4):\penalty0
  1250--1281, 2021.
\newblock URL \url{https://doi.org/10.1287/moor.2020.1094}.

\bibitem[Horst and Naujokat(2014)]{HorstNaujokat:14}
Ulrich Horst and Felix Naujokat.
\newblock When to cross the spread? {T}rading in {T}wo-{S}ided {L}imit {O}rder
  {B}ooks.
\newblock \emph{SIAM Journal on Financial Mathematics}, 5\penalty0
  (1):\penalty0 278--315, 2014.
\newblock URL \url{https://doi.org/10.1137/110849341}.

\bibitem[Huang et~al.(2019)Huang, Jaimungal, and
  Nourian]{HuangJaimungalNourian:19}
Xuancheng Huang, Sebastian Jaimungal, and Mojtaba Nourian.
\newblock Mean-field game strategies for optimal execution.
\newblock \emph{Applied Mathematical Finance}, 26\penalty0 (2):\penalty0
  153--185, 2019.
\newblock URL \url{https://doi.org/10.1080/1350486X.2019.1603183}.

\bibitem[Luo and Schied(2019)]{LuoSchied:20}
Xiangge Luo and Alexander Schied.
\newblock Nash equilibrium for risk-averse investors in a market impact game
  with transient price impact.
\newblock \emph{Market Microstructure and Liquidity}, 05\penalty0
  (01n04):\penalty0 2050001, 2019.
\newblock URL \url{https://doi.org/10.1142/S238262662050001X}.

\bibitem[Moallemi et~al.(2012)Moallemi, Park, and Roy]{MoallemiParkVanRoy:12}
Ciamac~C. Moallemi, Beomsoo Park, and Benjamin~Van Roy.
\newblock Strategic execution in the presence of an uninformed arbitrageur.
\newblock \emph{Journal of Financial Markets}, 15\penalty0 (4):\penalty0 361 --
  391, 2012.
\newblock ISSN 1386-4181.
\newblock URL
  \url{http://www.sciencedirect.com/science/article/pii/S138641811200002X}.

\bibitem[Naujokat and Westray(2011)]{NaujWes:11}
Felix Naujokat and Nicholas Westray.
\newblock Curve following in illiquid markets.
\newblock \emph{Mathematics and Financial Economics}, 4\penalty0 (4):\penalty0
  299--335, 2011.
\newblock ISSN 1862-9679.
\newblock URL \url{http://dx.doi.org/10.1007/s11579-011-0042-5}.

\bibitem[Neuman and Vo{\ss}(2021)]{NeumanVoss:21}
Eyal Neuman and Moritz Vo{\ss}.
\newblock Trading with the {C}rowd, 2021.
\newblock Preprint, available online at \url{https://arxiv.org/abs/2106.09267}.

\bibitem[Rogers and Singh(2010)]{RogerSin:10}
L.~C.~G. Rogers and Surbjeet Singh.
\newblock The cost of illiquidity and its effects on hedging.
\newblock \emph{Mathematical Finance}, 20\penalty0 (4):\penalty0 597--615,
  2010.
\newblock URL
  \url{https://onlinelibrary.wiley.com/doi/abs/10.1111/j.1467-9965.2010.00413.x}.

\bibitem[Schied and Zhang(2017)]{SchiedZhang:17}
Alexander Schied and Tao Zhang.
\newblock A state-constrained differential game arising in optimal portfolio
  liquidation.
\newblock \emph{Mathematical Finance}, 27\penalty0 (3):\penalty0 779--802,
  2017.
\newblock URL \url{https://onlinelibrary.wiley.com/doi/abs/10.1111/mafi.12108}.

\bibitem[Schied and Zhang(2019)]{SchiedZhang:19}
Alexander Schied and Tao Zhang.
\newblock A market impact game under transient price impact.
\newblock \emph{Mathematics of Operations Research}, 44\penalty0 (1):\penalty0
  102--121, 2019.
\newblock URL \url{https://doi.org/10.1287/moor.2017.0916}.

\bibitem[Schied et~al.(2017)Schied, Strehle, and Zhang]{SchiedStrehleZhang:17}
Alexander Schied, Elias Strehle, and Tao Zhang.
\newblock High-frequency limit of nash equilibria in a market impact game with
  transient price impact.
\newblock \emph{SIAM Journal on Financial Mathematics}, 8\penalty0
  (1):\penalty0 589--634, 2017.
\newblock URL \url{https://doi.org/10.1137/16M107030X}.

\bibitem[Sch{\"o}neborn(2008)]{Schoen:08}
Torsten Sch{\"o}neborn.
\newblock \emph{Trade execution in illiquid markets: {O}ptimal stochastic
  control and multi-agent equilibria}.
\newblock PhD thesis, Technische Universit{\"a}t Berlin, 2008.

\bibitem[Sch{\"o}neborn and Schied(2009)]{SchiedSchoeneborn:09}
Torsten Sch{\"o}neborn and Alexander Schied.
\newblock Liquidation in the {F}ace of {A}dversity: Stealth vs. {S}unshine
  {T}rading, 2009.
\newblock Preprint, available online at
  \url{https://papers.ssrn.com/sol3/papers.cfm?abstract_id=1007014}.

\bibitem[Strehle(2017)]{Strehle:17}
Elias Strehle.
\newblock Optimal execution in a multiplayer model of transient price impact.
\newblock \emph{Market Microstructure and Liquidity}, 3\penalty0 (4):\penalty0
  1850007, 2017.
\newblock URL \url{https://doi.org/10.1142/S2382626618500077}.

\end{thebibliography}

\end{document}